\documentclass[sigplan,10pt]{acmart}

\setcopyright{none}
\renewcommand\footnotetextcopyrightpermission[1]{}

\def\isReadyToSubmit{1}  
\def\doubleblind{0}      

\usepackage{xcolor}
\usepackage[normalem]{ulem}

\usepackage{adjustbox} 
\usepackage{algorithm} 
\usepackage[noend]{algpseudocode} 
\usepackage{anyfontsize} 
\usepackage{array} 
\usepackage{bm} 
\usepackage{calc} 
\usepackage{courier} 
\usepackage{dsfont} 
\usepackage{enotez} 
\setenotez{backref=true} 
\usepackage{flushend} 
\usepackage{ifthen} 
\usepackage{listings} 
\usepackage{multirow} 
\usepackage{makecell} 
\usepackage{paralist} 
\usepackage{pbox} 
\usepackage{pythonhighlight} 
\usepackage{soul} 
\usepackage{subcaption} 
\usepackage{tabularx} 
\usepackage{tikz} 
\usetikzlibrary{positioning, arrows.meta, shapes, shadows, matrix}
\usepackage{verbatim} 
\usepackage{wrapfig} 
\usepackage{xspace} 

\makeatletter
\renewcommand{\ALG@name}{Alg.}
\makeatother

\def\Def{Def.~}
\def\F{Fig.~}

\def\Thm{Thm.~}

\def\Lem{Lem.~}

\newcommand{\headingg}[1]{\noindent{\bf{#1}}} 
\newcommand{\heading}[1]{{\vspace{1pt}\noindent\bf{#1}}} 
\newcommand{\headingi}[1]{\noindent{\vspace{1pt}\em{#1}}} 

\algblock{Input}{EndInput}
\algnotext{EndInput}
\algblock{Output}{EndOutput}
\algnotext{EndOutput}
\algblock{Config}{EndConfig}
\algnotext{EndConfig}

\newcommand{\graycomment}[1]{\textcolor{gray}{// \em #1}}
\algrenewcommand\algorithmicindent{0.5em} 

\newcommand{\code}[1]{\lstinline|#1|}

\newcommand{\circled}[1]{%
  \tikz[baseline=(char.base)] \node[draw, circle, inner sep=1pt](char) {#1};%
}
\newcommand{\step}[2][blue]{{\color{#1} \circled{#2}}}

\ifnum\isReadyToSubmit=0
\newcommand{\asaf}[1]{\textcolor{orange}{\{Asaf: #1\}}}
\newcommand{\ml}[1]{\textcolor{purple}{\{Mathias: #1\}}}
\newcommand{\pierre}[1]{{\color{blue}{Pierre: #1}}}
\newcommand{\pierrerm}[1]{{\color{blue}{\st{#1}}}}
\newcommand{\pierrerp}[2]{{\color{blue}{\st{#1}{#2}}}}
\newcommand{\ben}[1]{{\color{green}{Ben: #1}}}
\newcommand{\roxana}[1]{{\color{cyan}{Roxana: #1}}}
\newcommand{\giorgio}[1]{{\color{green!70!black}{Giorgio: #1}}}
\newcommand{\todo}[1]{{\color{red} TODO(#1). }}
\newcommand{\todomsg}[2]{{\color{red} TODO(#1): #2}}
\newcommand{\mc}[1]{{\color{orange}{Mark: #1}}}
\newcommand{\mcrm}[1]{{\color{orange}{\st{#1}}}}
\newcommand{\mcrp}[2]{{\color{orange}{\st{#1}{#2}}}}
\newcommand{\ac}[1]{{\color{pink}{Alison: #1}}}
\newcommand{\tosupport}[1]{{\uline{#1}}}
\newcommand{\jason}[1]{\textcolor{red}{JN: #1}}
\newcommand{\xiaotian}[1]{\textcolor{teal}{\{Xiaotian: #1\}}}
\newcommand{\jo}[1]{\textcolor{green!40!black}{\textbf{JO:} #1}}
\newcommand{\shuren}[1]{\textcolor{magenta}{[Shuren: #1]}}
\newcommand{\navid}[1]{\textcolor{teal}{[Navid: #1]}}

\else  

\newcommand{\asaf}[1]{}
\newcommand{\ml}[1]{}
\newcommand{\pierre}[1]{}
\newcommand{\ben}[1]{}
\newcommand{\kelly}[1]{}
\newcommand{\roxana}[1]{}
\newcommand{\pierrerm}[1]{}
\newcommand{\pierrerp}[2]{}
\newcommand{\todo}[1]{}
\newcommand{\todomsg}[2]{}
\newcommand{\mc}[1]{}
\newcommand{\mcrm}[1]{}
\newcommand{\mcrp}[2]{}
\newcommand{\ac}[1]{}
\newcommand{\giorgio}[1]{}
\newcommand{\jason}[1]{}
\newcommand{\xiaotian}[1]{}
\newcommand{\jo}[1]{}
\newcommand{\shuren}[1]{}
\newcommand{\navid}[1]{}
\newcommand{\tosupport}[1]{}
\renewcommand{\endnote}[1]{}
\fi  

\newtheoremstyle{TheoremNum}
{\topsep}{\topsep}              
{\itshape}                      
{}                              
{\bfseries}                     
{.}                             
{ }                             
{\thmname{#1}\thmnote{ \bfseries #3}}
\theoremstyle{TheoremNum}

{\makeatletter
 \gdef\xxxmark{%
   \expandafter\ifx\csname @mpargs\endcsname\relax
     \expandafter\ifx\csname @captype\endcsname\relax
       \marginpar{\textcolor{red}{xxx~}}%
     \else
       \textcolor{red}{xxx~}%
     \fi
   \else
     \textcolor{red}{xxx~}%
   \fi}
 \gdef\xxx{\protect\xxx@inner}
 \gdef\xxx@inner{\@ifnextchar[\xxx@lab\xxx@nolab}
 \long\gdef\xxx@lab[#1]#2{{\bf [\xxxmark \textcolor{red}{#2} ---{\sc #1}]}}
 \long\gdef\xxx@nolab#1{{\bf [\xxxmark \textcolor{red}{#1}]}}
 \ifnum\isReadyToSubmit=1
   \long\gdef\xxx@lab[#1]#2{}\long\gdef\xxx@nolab#1{}%
 \fi
}

\makeatletter
\renewcommand\section{\def\@toclevel{1}%
  \@startsection{section}{1}{\z@}%
  {-4\p@ \@plus -1\p@ \@minus -1\p@}%
  {2\p@}%
  {\ACM@NRadjust\@secfont}}
\renewcommand\subsection{\def\@toclevel{2}%
  \@startsection{subsection}{2}{\z@}%
  {-4\p@ \@plus -1\p@ \@minus -1\p@}%
  {1\p@}%
  {\ACM@NRadjust\@subsecfont}}
\renewcommand\subsubsection{\def\@toclevel{3}%
  \@startsection{subsubsection}{3}{\z@}%
  {-1\p@ \@plus -1\p@ \@minus -1\p@}%
  {-3.5\p@}%
  {\ACM@NRadjust{\@subsubsecfont\@adddotafter}}}
\let\ACM@origsection\section
\let\ACM@origsubsection\subsection
\let\ACM@origsubsubsection\subsubsection
\makeatother

\newcommand{\numAppsTotal}{five}
\newcommand{\numAppsUntrusted}{four}

\newcommand{\maxEC}{\ensuremath{\mathtt{maxEC}}}
\newcommand{\maxECApp}{\ensuremath{\mathtt{maxEC_{app}}}}
\newcommand{\maxECSys}{\ensuremath{\mathtt{maxEC_{sys}}}}
\newcommand{\minEC}{\ensuremath{\mathtt{minEC}}}

\newcommand{\cL}{\mathcal{L}}
\newcommand{\cT}{\mathcal{T}}
\newcommand{\cE}{\mathcal{E}}

\newcommand{\cD}{\mathcal{D}}

\newcommand{\cX}{\mathcal{X}}
\newcommand{\cY}{\mathcal{Y}}

\newcommand{\cB}{\mathcal{B}}
\newcommand{\R}{\mathbb{R}}

\newcommand{\bfy}{\mathbf{y}}

\ifnum\doubleblind=1
  \newcommand{\sysname}{SensOS\xspace}  
\else  
  \newcommand{\sysname}{CityOS\xspace}
\fi  

\title{\sysname: Privacy Architecture for Urban Sensing}

\ifnum\doubleblind=1
  \acmSubmissionID{406}
  
  \author{Paper \#406}
  \affiliation{\country{}}
\else  
  \author{Giorgio Cavicchioli}
  \affiliation{\institution{Columbia University}\country{USA}}

  \author{Mark Chen}
  \affiliation{\institution{Columbia University}\country{USA}}

  \author{Navid Salami Pargoo}
  \affiliation{\institution{Rutgers University}\country{USA}}

  \author{Shuren Xia}
  \affiliation{\institution{Rutgers University}\country{USA}}
  
  \author{Xiaotian Zhou}
  \affiliation{\institution{Rutgers University}\country{USA}}

  \author{Roxana Geambasu}
  \authornote{Also currently partly with Google, formerly partly with Meta. Work done exclusively in the context of Columbia research.}
  \affiliation{\institution{Columbia University}\country{USA}}

  \author{Jason Nieh}
  \affiliation{\institution{Columbia University}\country{USA}}

  \author{Jorge Ortiz}
  \affiliation{\institution{Rutgers University}\country{USA}}
\fi  

\begin{abstract}

Cities are rapidly deploying sensing infrastructure -- cameras, environmental sensors, and connected kiosks -- that continuously observe public spaces, yet they lack a system architecture governing how applications access, aggregate, and retain this data, creating privacy risks and preventing consistent policy enforcement. We present {\em \sysname}, an operating system for urban sensing that mediates application access to sensor data through a three-tier API inspired by structured, privacy-conscious web interfaces. The tiers expand the spatial scope of data access while imposing progressively stronger privacy constraints: {\em On-Scene} supports real-time sensing with raw data confined to the local context; {\em Single-Locality Aggregation} enables differentially private longitudinal statistics at a fixed location; and {\em Cross-Locality Aggregation} supports citywide analytics via aggregation across locations, with user devices enforcing per-user privacy budgets. \sysname runs as an edge runtime that executes untrusted applications in ephemeral containers, enforcing these policies and providing transparency via broadcasts of differential privacy loss. We implement \sysname and applications across all tiers -- including pedestrian safety alerts, real-time and forecast parking availability, traffic dashboards, and subway trajectory measurement -- and show that it supports practical streetscape applications while enforcing strong privacy.

\end{abstract}

\begin{document}

\maketitle
\pagestyle{plain}

\vspace{-0.25cm}
\section{Introduction}
\label{sec:introduction}

Cameras, audio sensors, and other sensing infrastructure are increasingly deployed across public and semi-public urban spaces, yet their data-access practices are governed by few consistent rules. New York City illustrates the problem: LinkNYC kiosks have been found to violate privacy policies, including failures to anonymize device identifiers~\cite{nyclu-linknyc-2023, intercept-linknyc-2018}; Wegmans collects facial-recognition data without disclosing retention practices~\cite{gothamist-wegmans-2026, cnn-wegmans-2026}; and CitiBike ride data can be re-identified at high rates using minimal auxiliary information~\cite{schneider-citibike-2016}.
These patterns extend beyond a single city. Deployments such as Singapore's Smart Nation and Barcelona's sensor-equipped smart poles illustrate the global expansion of urban sensing infrastructure~\cite{ganz-smartcity-2024}. While some high-profile projects have failed amid privacy concerns -- Sidewalk Labs in Toronto~\cite{statescoop-sidewalk-2020} -- many incremental deployments continue to grow with limited oversight or consistent governance.

The core risk is not any single deployment but the trajectory: without structured interfaces governing how sensing data is accessed and aggregated, backend collection becomes the default and, once entrenched, is difficult to reverse. The web illustrates this pattern. For decades, advertising measurement relied on cross-site tracking -- third-party cookies, fingerprinting, and backend data exchanges -- largely because browsers offered no structured alternative. Recent efforts, such as the W3C Attribution API~\cite{w3c-attribution-2024}, embed privacy constraints into platform interfaces via encrypted reporting and differentially private aggregation. Whether such approaches can displace tracking remains uncertain, but their difficulty underscores the importance of early intervention.

Physical-world sensing faces a similar inflection point. Pervasive cross-context tracking is not yet the norm in most environments, and some legal constraints have begun to emerge~\cite{nyc-biometric-law-2021}. This creates an opportunity to define structured data-access models before tracking becomes entrenched.

We introduce such a framework. Inspired by the web's recent move
toward structured, privacy-conscious interfaces, we present {\em \sysname}, a
privacy architecture for urban sensing that separates data access into
distinct classes, each exposed through a purpose-built API with privacy
protections calibrated to its risk.
\sysname is a privacy-enforcing runtime deployed at the edge  --  at intersections or other instrumented locations  --  where it hosts untrusted applications in sandboxed containers and mediates all access to local sensor data through a three-tier API.
{\em API~1 (On-Scene)} provides real-time, ephemeral access with outputs confined to the immediate locality  --  reflecting that momentary, local observation is qualitatively different from persistent recording. {\em API~2 (Single-Locality Aggregation)} supports longitudinal statistics at a single location through differential privacy (DP)  --  addressing the risk that even single-site data accumulates into a sensitive behavioral picture over time. {\em API~3 (Cross-Locality Aggregation)} enables citywide analysis mediated through user devices with per-user privacy budgets  --  ensuring that cross-location measurement, the most surveillance-prone capability, never relies on infrastructure tracking individuals directly. This tiered structure was developed by systematically mapping levels of data access identified on the web to their physical-world counterparts  --  an analogy we develop in \S\ref{sec:analogy-main-section} and list among our key contributions.

As a cross-cutting design, \sysname provides \emph{privacy transparency}. Each node wirelessly broadcasts its DP loss, enabling nearby devices running a \sysname mobile app to track cumulative privacy exposure as users move through instrumented spaces. This per-user accounting -- combining losses from API~2 broadcasts and API~3 participation -- gives individuals a continuously updated measure of their total privacy expenditure across all \sysname-instrumented locations.

A critical question is whether this web-inspired privacy architecture can support useful smart-city applications. We prototype \numAppsTotal applications: (1) an object detection and tracking service that monitors the local scene and disseminates bounding boxes to nearby consumers (e.g., autonomous vehicles) to improve safety (API~1); (2) a pedestrian safety app built on (1) that detects imminent pedestrian-vehicle collisions and triggers audio alerts (API~1); 
(3) a parking app that provides real-time availability locally (API~1) and model-based predictions for distant locations (API~2); (4) a traffic dashboard reporting per-locality pedestrian and vehicle statistics (API~2); and (5) a method for estimating subway route popularity -- typically requiring tap-in/tap-out systems, but absent in NYC -- without tracking people (API~3).

Overall, we make the following contributions: a web-to-smart-spaces privacy analogy that maps canonical web data-access levels -- in-session use, first-party longitudinal data, and cross-site data -- to their physical-world counterparts, providing a reference model that motivates the API structure and grounds its privacy protections in a platform that has confronted tracking at scale (\S\ref{sec:analogy-main-section}); the design of a structured API for urban sensing, based on this analogy, that separates functionality into three tiers with risk-tailored privacy protections (\S\ref{sec:overview}); new system mechanisms enabling this architecture, namely
ephemeral containers, 
localized release,
ephemeral DP continual release,
bounded tracking contexts,
relevance self-identification, and
device-side privacy accounting
-- a subset of which address key gaps in the analogy (\S\ref{sec:detailed-design}); a prototype implementation and \numAppsTotal representative applications demonstrating support for practical smart-space functionality (\S\ref{sec:apps}); and an experimental evaluation showing that \sysname provides good utility for these applications while enforcing privacy (\S\ref{sec:evaluation}). We release all code.

\section{Web-to-Smart-Spaces Privacy Analogy}
\label{sec:analogy-main-section}

We use the web as a reference model for designing privacy
in urban sensing: its trajectory shows how tracking becomes
entrenched in the absence of structured support, and how
platforms can avoid this through purpose-built,
privacy-preserving interfaces.

\subsection{Background on Web Privacy}
\label{sec:web-privacy-background}

Browsers were designed around origin-based isolation: the
{\em same-origin policy} restricts how content from one site
can access data from another. But advertising measurement --
one of the web's dominant workloads -- is inherently {\em
cross-origin}, requiring linking ad impressions on content
sites to conversions on advertiser sites. Because browsers
offered no structured API for this task, the ecosystem defaulted
to tracking. Third-party cookies enabled persistent identifiers
across sites; when some browsers blocked them, tracking shifted
to fingerprinting and backend data exchanges, such as first-party
cookie sharing. The pattern is general: when platforms lack
structured support for a required function, applications
reconstruct it through tracking.

The web is now attempting to correct this through a combination
of regulation, platform interventions, and standards. Mechanisms
such as Apple's Private Relay~\cite{apple-private-relay} and browser
fingerprinting resistance~\cite{disable-remote-fingerprinting-apple,disable-remote-fingerprinting-mozilla}
introduce technical barriers, while regulation (e.g., GDPR~\cite{gdpr-2016})
and platform responses such as Apple's ATT~\cite{att} require tracking
to be surfaced and permissioned. New U.S.\ state laws~\cite{ccpa-2020, CTDPA-2022, NJDPA-2024} further grant users the right to opt out of backend
data sharing; this right is now being operationalized via W3C's Global
Privacy Control~\cite{w3c-gpc-2024,hausladenEtAlGPCWeb2025}, making it increasingly enforceable. These efforts create resistance to tracking but,
without a structured alternative for advertising measurement,
are unlikely to suffice.

The W3C's Attribution API draft standard~\cite{w3c-attribution-2024,
TKM+24} introduces such an alternative: a structured browser
interface for privacy-preserving advertising measurement.
Content sites register ad impressions with the browser; advertiser
sites request encrypted reports linking conversions to prior
impressions; and only DP aggregates across many users are released
via secure multiparty computation or trusted execution environments.
Differential privacy (DP) bounds the influence of any individual on
an output, with privacy budgets limiting cumulative exposure over
time~\cite{dwork-roth-2014}. In Attribution, privacy is enforced
on-device via per-epoch budgets, an enforcement architecture we 
ourselves developed in~\cite{TKM+24}.

In essence, Attribution provides a {\em cross-origin interface}
with privacy guarantees that approximate {\em same-origin isolation}.
Whether it will ultimately displace tracking remains uncertain,
but it demonstrates that viable alternatives exist. The broader
lesson is clear: without structured interfaces, tracking becomes
the default and hard to reverse. We apply this lesson to urban
sensing by designing structured, privacy-by-design interfaces
for different data-access types.

\subsection{Why the Web as a Reference}
\label{sec:why-use-the-web-for-analogy}

Designing a privacy architecture for a new domain without a reference model is extremely challenging. What should the data-access model be? What privacy policies are appropriate? How should one navigate the inevitable tradeoffs between privacy and functionality? And how can these choices be explained to diverse stakeholders to achieve acceptance and adoption? These questions are difficult not only technically, but also from the perspective of building consensus around a new system design.

These challenges are particularly acute in smart cities and public-space sensing. Such systems differ fundamentally from traditional platforms: they passively observe people who may not carry devices, do not establish sessions, and cannot be prompted for consent. As a result, the notice-and-consent model~\cite{oecd-privacy-guidelines-1980}, common in traditional platforms, does not readily apply. Even when attempted~\cite{PriPref,IPic,PrivacyTag,DoNotShare}, it is widely recognized as insufficient~\cite{solove-consent-2013}. Modern platforms such as browsers are therefore beginning to move away from notice-and-consent by embedding privacy protections directly into their interfaces (e.g., Attribution). In the absence of a widely accepted reference model, however, it remains difficult to develop credible proposals, justify tradeoffs, or build agreement across stakeholders.

We propose the web as that reference. Despite differences in interaction models, functionality, and data types they handle -- the web handles virtual activity while smart spaces handle physical activity -- both domains share a common structure: applications rely on a small set of canonical data-access patterns with distinct privacy risks. This makes the web a useful foundation for designing structured interfaces for urban sensing, allowing us to transfer its lessons -- both good and bad -- to inform a principled design that avoids repeating its mistakes.

\subsection{Web-to-Smart-Spaces Analogy}
\label{sec:web-to-smart-cities-analogy}

Our analogy is based on a common set of canonical data-access
patterns, each with distinct privacy risks. These range from
in-session, single-origin access to cross-origin access on the
web, and from on-scene sensing at a single location to
cross-location trajectory data in smart spaces. Broader-scope
access enables richer functionality but also increases privacy
risk, requiring stronger controls. \sysname structures data
access and privacy around these patterns.

\heading{Level 1: In-context, short-term data use.} 
On the web, this is real-time first-party use of
in-session data (e.g., a user's current query or page activity),
which browsers largely permit. In smart spaces, the analogue is
real-time, on-scene access to ephemeral sensor data for immediate
functions such as obstacle detection, safety alerts, or parking
availability. In both cases, data is used in context and
need not be retained or linked across contexts, resulting in
low and intuitive privacy risk. When visiting a website, you
expect the site to observe your actions within that session; 
what you may not expect is them to retain this data or share it
with other sites. Similarly, when entering an instrumented
public space, you expect momentary observation for local
functionality, but not persistent recording or tracking across
the city. \sysname mirrors this expectation at this tier
by ensuring that data and outputs are ephemeral and confined to the
local context.

\heading{Level 2: Single-context, longitudinal aggregation.}
On the web, first-party sites collect user activity over time
(e.g., visits, interactions, preferences), typically via
first-party cookies. While useful for personalization, this
accumulation produces increasingly sensitive behavioral profiles
yet it is largely unconstrained by browsers. In smart spaces, the
analogue is single-locality aggregation: collecting data over
time at a fixed location (pedestrian counts, traffic
statistics). Although data remains local, accumulation increases
privacy risk. The web shows that, left unchecked, parties combine
data across contexts, leading to progressively broader collection.
Here, \sysname diverges: rather than relying on self-regulation,
it applies DP to bound each individual's contribution, limiting
longitudinal inference and the impact of future compromise.

\heading{Level 3: Cross-context aggregation.}
On the web, the most sensitive functionality is cross-site
tracking, historically implemented via third-party cookies or
fingerprinting. Recent efforts such as Attribution replace
tracking with privacy-preserving measurement. In
smart spaces, the analogue is cross-locality measurement:
linking observations across locations to estimate flows,
trajectories, or system-wide patterns. While valuable, this
poses the highest privacy risk by enabling reconstruction of
individual movement. \sysname adopts and adapts Attribution to
the physical world: rather than infrastructure tracking
individuals, cross-locality queries are mediated via user
devices, which generate encrypted reports and enforce per-user
privacy budgets. 

\section{\sysname Overview}
\label{sec:overview}

Referencing the web, \sysname organizes data access into three levels with distinct privacy risks, in which broader access along some dimensions -- spatial scope, temporal scope, or fidelity -- is balanced by tighter constraints on the outputs derived from that access.
{\em API~1 (On-Scene)} provides the richest access
to raw, high-fidelity data, but strictly confines computation
and outputs to the immediate physical and temporal context in
which the data is produced. {\em API~2 (Single-Location
Aggregation)} relaxes temporal locality by allowing aggregation
over time within a single location, while restricting access
to aggregate statistics rather than raw data.
{\em API~3 (Cross-Location Aggregation)} further relaxes spatial
locality of the data by supporting cross-location measurement,
but only through a constrained interface in which the user's
device controls their privacy exposure.

This progression mirrors the web's evolution from in-session,
single-origin access to cross-origin aggregation: as the scope
of data use expands, the platform imposes stronger structural
constraints to limit what is revealed. In this sense, \sysname
maps geographic locality in physical space to origin locality
on the web, giving rise to a {\em same-location policy} -- a
physical analog of the same-origin policy -- under which data
and computation remain local by default, with structured
exceptions for broader access. API~2 occupies an intermediate
point that has no direct analogue on the web -- a gap we view
as a design flaw that enables unbounded accumulation of
first-party data and that we avoid in physical sensing.


\subsection{API Example}
\label{sec:example}

Consider an organization seeking to improve streetscape safety
using video from cameras already deployed at city intersections.
It plans three functions: (1) real-time alerts warning pedestrians
and cyclists of approaching vehicles; (2) a dashboard showing
hourly pedestrian and vehicle counts at each intersection; and
(3) a citywide analysis of pedestrian flows between neighborhoods.
Without architectural constraints, the natural approach is to
stream camera feeds to a backend, store them, and run all three
functions on the collected data. The organization might apply
basic anonymization, such as blurring faces, but such techniques
are fragile: individuals can often be re-identified through gait,
body shape, clothing, or spatiotemporal patterns, especially when
data is retained or linked across locations.

\sysname restructures this. Real-time alerts run on-scene via
API~1: the application processes the local feed in an ephemeral
container, emits warnings to nearby devices, and retains nothing.
The per-intersection dashboard uses API~2: the node releases
DP hourly counts, never exposing individual
observations. Citywide flow analysis uses API~3: rather than
tracking pedestrians across intersections, participating users'
devices contribute encrypted reports that are aggregated without
revealing individual trajectories. The same functionality is
supported, but no persistent, city-scale record of individual
activity is ever created -- and privacy does not rest on the
hope that anonymization holds up.

\subsection{Threat Model}
\label{sec:threat-model}

\headingg{Deployment model.}
\sysname targets deployments in which
\begin{wrapfigure}{r}{0.45\columnwidth}
    \centering
    \vspace{-0pt}
    \includegraphics[width=0.44\columnwidth]{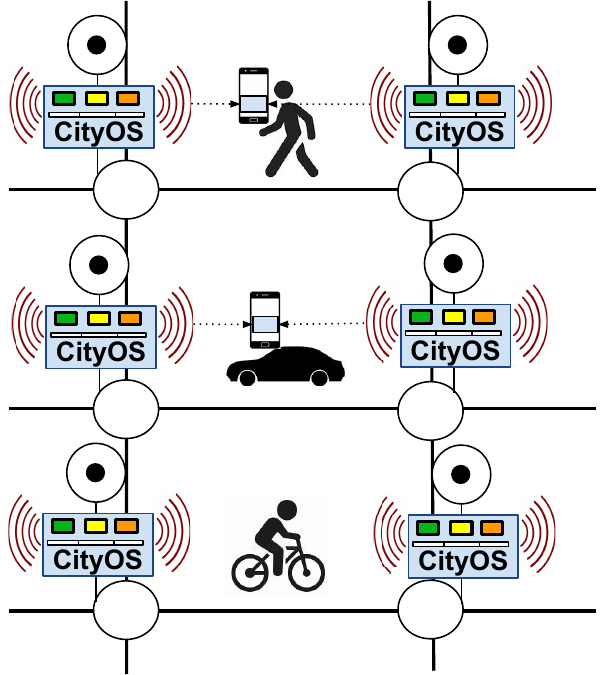}
    \caption{\bf Deployment model.}
    \label{fig:smart-city-integration}
    \vspace{-2pt}
\end{wrapfigure}
 sensing and computation occur at many
distributed locations (e.g., intersections, parking lots, transit stops, or
retail spaces). 
\F\ref{fig:smart-city-integration} illustrates this setting:
nodes are placed at intersections, each equipped with sensors (e.g., cameras)
and hosting one or more applications. People move through the environment,
some carrying mobile devices, others not.

\heading{Trust assumptions.}
\emph{Nodes} are trusted runtimes that mediate all application access to
local sensors and enforce API restrictions. Where possible, nodes at
different locations are independently administered so that compromise of
one location does not extend to others. When nodes are operated by a single
administrative domain, we assume -- enforced via policy -- that they do not
exchange data outside the system interfaces. Accordingly, each node is
trusted with data from its own location, but not with data from others.
However, nodes may be temporarily compromised; we minimize retained state
to limit exposure but do not defend against advanced persistent threats.

\emph{Applications} run on nodes and implement sensing functionality. We do
not trust their data practices and assume that, without enforced structure,
they would default to backend data collection. The architecture therefore
bounds what applications can observe or export through sandboxing and
API restrictions. Applications may additionally be subject to policy
or contractual constraints imposed by the operator.

\emph{Mobile devices} receive localized outputs broadcast by nearby nodes
and participate in privacy-budgeted aggregation for cross-locality queries.
Each user trusts their own device but not others'. We assume that most
devices do not run software designed to systematically collect and
exfiltrate nearby nodes' outputs -- in effect, that what is broadcast at an
intersection largely remains there. This assumption may be reinforced
through policy and regulation (\S\ref{sec:api1}).

\heading{Role of policy.}
\sysname enforces structural constraints technically (brief retention,
locality-scoped outputs, structured aggregation APIs), but as on the web,
effective privacy also requires legal or organizational constraints. Our
goal is to enforce as much as possible technically while keeping the
remaining policy assumptions explicit and narrow.

\heading{Limits of anonymization.}
Data cannot be fully anonymized and remain useful~\cite{dwork-roth-2014},
and individuals can often be re-identified through gait, body shape,
clothing, or spatiotemporal patterns.
We reject anonymization as justification for unrestricted backend
collection and instead make structural constraints the default (\S\ref{sec:api1}).

\setlength{\fboxsep}{0pt}      
\setlength{\fboxrule}{0.2pt}   

\begin{figure*}[t]
\centering

\begin{subfigure}[t]{0.63\linewidth}
    \vspace{0pt}
    \centering
    \fbox{\includegraphics[width=\linewidth]{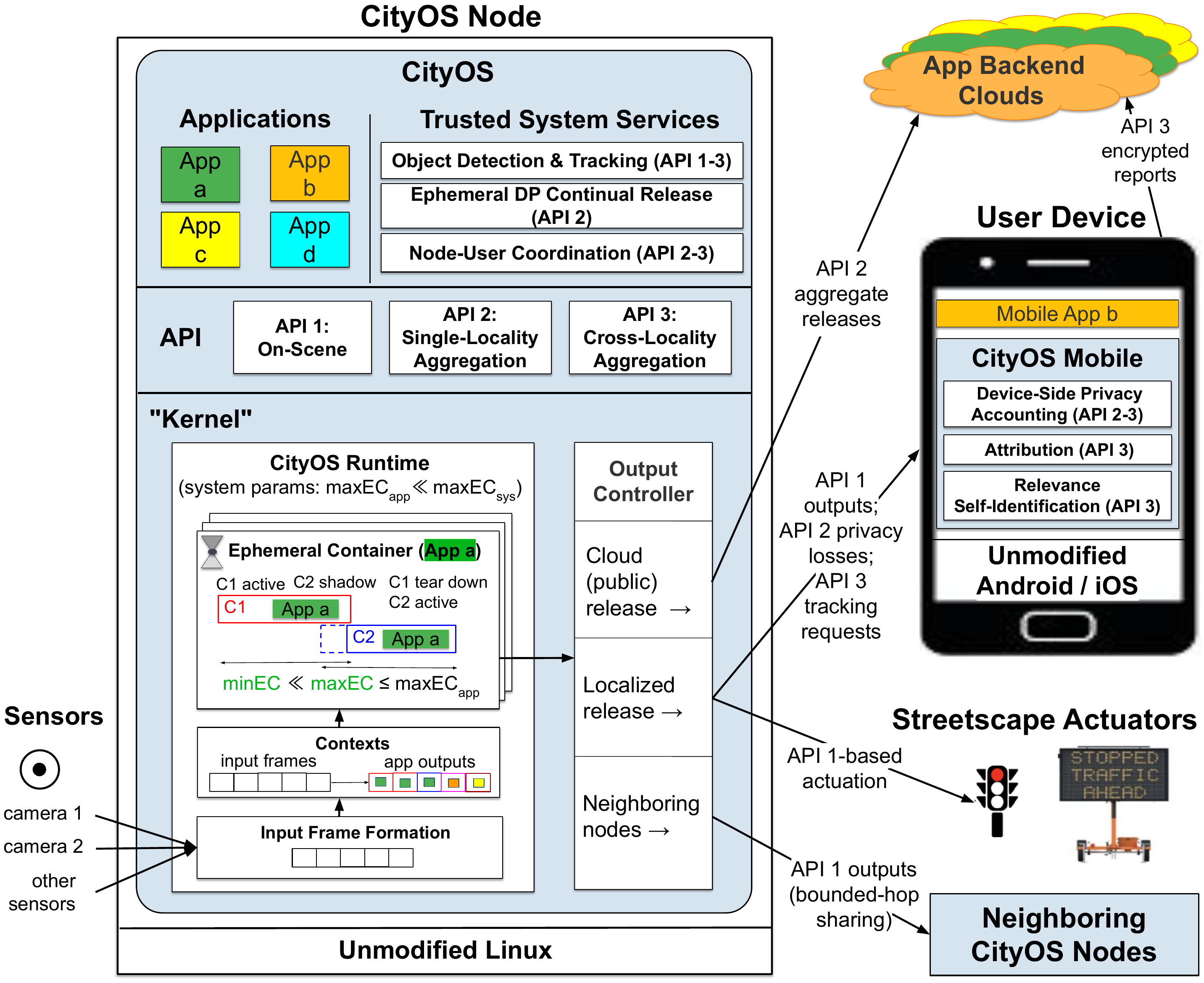}}
    \caption{\sysname architecture}
    \label{fig:architecture}
\end{subfigure}
\hfill
\begin{minipage}[t]{0.33\linewidth}
    \vspace{0pt}

    \begin{subfigure}[t]{\linewidth}
        \vspace{0pt}
        \centering
        \fbox{\includegraphics[width=\linewidth]{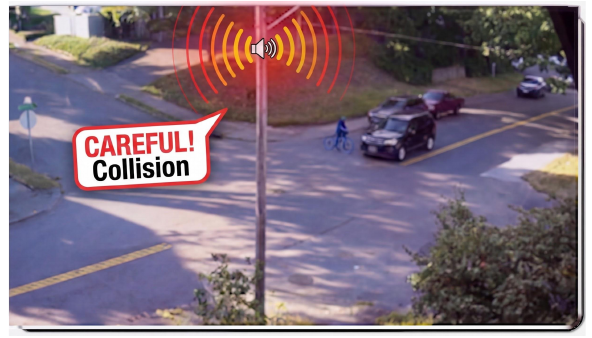}}
        \caption{Pedestrian safety app (visual intuition).}
        \label{fig:apps:ped-safety}
    \end{subfigure}

    \vspace{4em}

    \begin{subfigure}[t]{\linewidth}
        \vspace{0pt}
        \centering
        \fbox{\includegraphics[width=\linewidth]{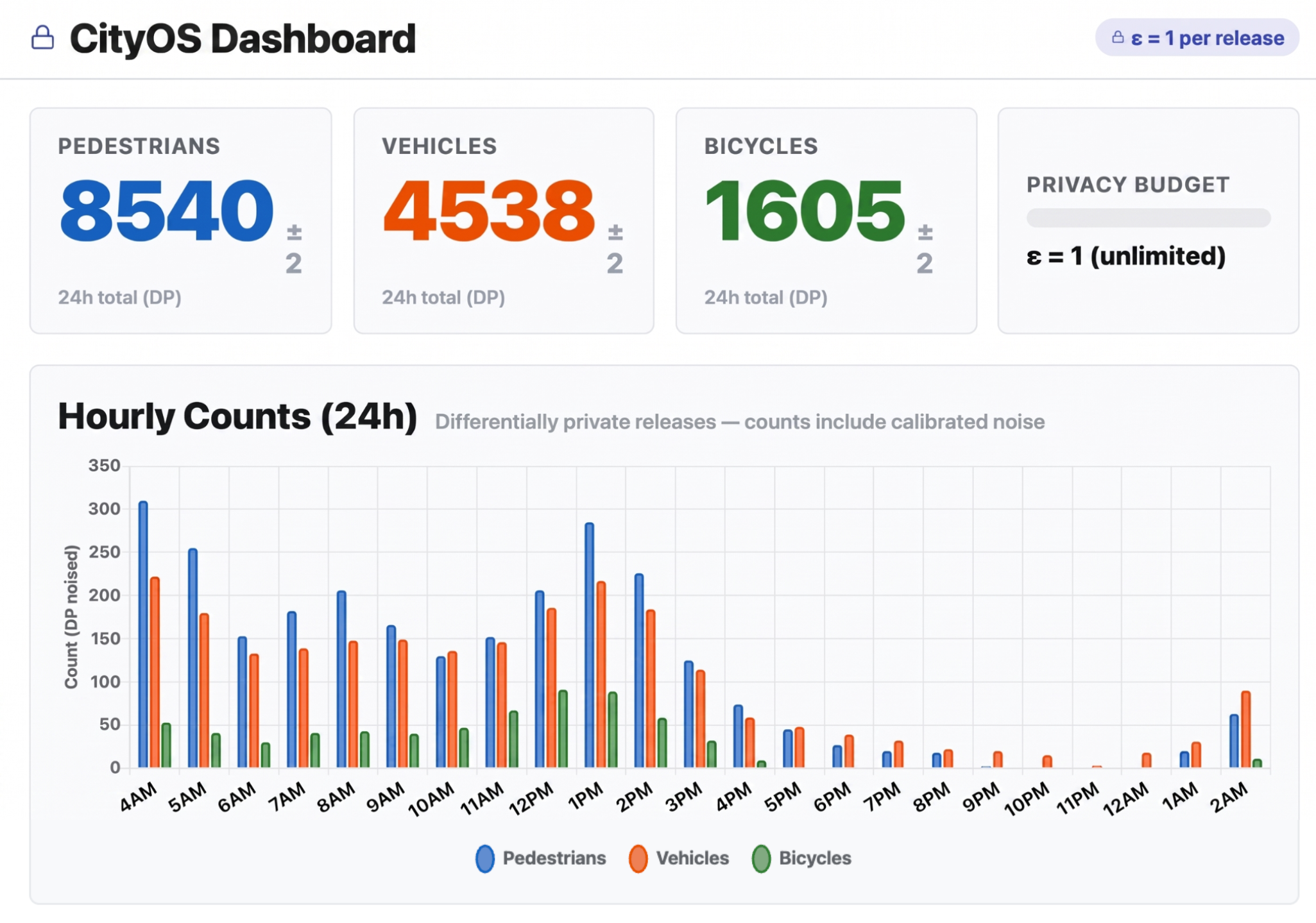}}
        \caption{Dashboard app (actual screenshot).
        }
        \label{fig:apps:dashboard}
    \end{subfigure}
\end{minipage}

\vspace{1em}
\caption{{\bf (a) \sysname architecture; (b), (c) running-example applications.}}
\label{fig:sysname-overview}
\end{figure*}

\heading{Scope.}
We do not aim to defend against all forms of surveillance. Broad government
surveillance is out of scope, including adversaries capable of tapping many
nodes or compelling broad access to system state. Similarly, large companies
with ubiquitous access across mobile devices are assumed to be disincentivized
from large-scale surveillance when privacy-preserving APIs are available,
in part by regulation and consumer pressure. Compromise of a single node may
still be harmful, particularly at sensitive locations. Our claim is not that
\sysname eliminates surveillance risk, but that it can prevent tracking and
backend collection from becoming the default practice for physical-space
applications.

We assume attackers may compromise some applications,
nodes, or mobile devices, but cannot maintain persistent visibility
across a large fraction of locations for extended periods. Under
this assumption, \sysname limits the scope and retrospective
value of compromise by minimizing raw-data retention, restricting
infrastructure-based tracking and mediating broader measurements
with private interfaces.

\subsection{Architecture}
\label{sec:architecture}

\heading{Design principles.}
In designing \sysname, we follow several key principles derived
from the web-based same-location analogy: (1) functionality that
requires raw sensing, and the outputs it produces, remains local
to where it occurs; (2) functionality that requires information
beyond a single moment relies on aggregation interfaces that
bound what individual observations reveal; and (3) functionality
that requires cross-location information is mediated by user
devices rather than by infrastructure tracking individuals
across space. In addition, we adopt a data-minimization principle
from our threat model: (4) both applications and the system
itself minimize state accumulation, so that even if a node is
temporarily compromised, the information exposed remains limited.

\heading{Node-side architecture.}
\F\ref{fig:architecture} shows the \sysname architecture.
\sysname is structured like a traditional OS, with a kernel,
system services, and an API layer, but everything runs in
user space over unmodified Linux. Applications run in
sandboxed containers and access sensor data exclusively
through the three-tier API.
System services -- including Object Detection, DP Continual
Release, and Node-User Coordination -- are trusted components
that implement parts of the API. They operate atop the kernel's
abstractions and within constraints similar to untrusted
applications, but with selectively expanded privileges as needed.
The kernel provides the execution environment and I/O substrate:
it assembles raw sensor data into \emph{input frames} (the basic
input unit) and groups them into temporally bounded
\emph{contexts}, the unit at which \sysname enforces ephemerality
and output localization. An Output Controller governs all data
leaving the node, routing outputs based on the API tier that
produced them.

\heading{User device side.}
Some user devices run a \sysname mobile app alongside ordinary
apps on Android or iOS. The \sysname app incorporates: an
open-source implementation of Attribution~\cite{pdslib} for
privacy-budgeted API~3 participation; a
Relevance Self-identification module that communicates with
the node to determine whether the user is relevant to ongoing
API~2 or API~3 measurements at a location; and a Privacy Loss
Accounting module that accumulates API~2 and API~3 losses.
In addition to this core app, each \sysname application may have
companion mobile apps that leverage locally emitted
signals. For example, a parking application may query nearby
nodes for live availability and notify users when a spot is
available near their destination (\S\ref{sec:apps}).

\subsection{Execution Example}
\label{sec:execution-example}

We illustrate \sysname's key mechanisms by tracing the three
functions from \S\ref{sec:example} through the architecture.
\S\ref{sec:apps} describes our own implementations of these
functions on \sysname.

The real-time pedestrian alert runs as an API~1 application
inside an \emph{ephemeral container}. To the application, this
appears as a continuously running container with access to
sensor frames, but \sysname restricts this view to a sliding
window of the last $\minEC$ frames. This provides sufficient
temporal context to detect events such as collision courses,
which cannot be inferred from a single frame. More precisely,
\sysname guarantees that the container always holds at least
$\minEC$ frames and never more than $\maxEC$ (older state is
automatically expunged). The container emits outputs -- here,
collision warnings -- through the Output Controller's localized
release, either to nearby mobile devices or to a connected alert
device. \F\ref{fig:apps:ped-safety} visualizes our implementation.

The per-intersection dashboard runs as an API~2 application.
Each ephemeral container emits hourly pedestrian and vehicle
counts into an aggregation window, whose duration is
application-specified. At the end of each window, the kernel
clamps each container's contribution to a bounded sensitivity,
sums the clamped values, and adds calibrated noise to produce a
DP release -- the only API~2 output that
leaves the node, forwarded via cloud release to the backend.
The associated privacy loss is broadcast locally so nearby devices can update their privacy loss accounting. \F\ref{fig:apps:dashboard} 
gives a screenshot of our implementation.

Finally, citywide flow measurement uses API~3, implemented primarily
on the mobile side. The node's User-Device Coordination service
emits annotated bounding boxes tagged with tracking or reporting
requests. Nearby devices use the Node-User Coordination module on their side to determine whether they are relevant to the measurement. E.g., if the analysis concerns pedestrians, 
only devices carried by pedestrians self-identify as relevant,
track locally, and later report encrypted, privacy-budgeted 
trajectories following W3C Attribution. The organization receives
only aggregate flow statistics. \S\ref{sec:apps} describes a related 
application we built on \sysname.

\section{Design}
\label{sec:detailed-design}

\F\ref{fig:detailed-design} shows the detailed architectures of
APIs~1 and~2; API~3 is illustrated in \F\ref{fig:api3} in the
appendix, though our description does not depend on it. We
describe each API in turn, focusing on the key novel systems
mechanisms needed to implement them. Security, privacy, and
utility analyses appear in the appendix (\S\ref{appendix:api1},
\S\ref{appendix:api2}, and \S\ref{appendix:api3} for APIs~1, ~2, and~3, respectively).

\begin{figure*}[t]
    \centering
    \includegraphics[width=\linewidth]{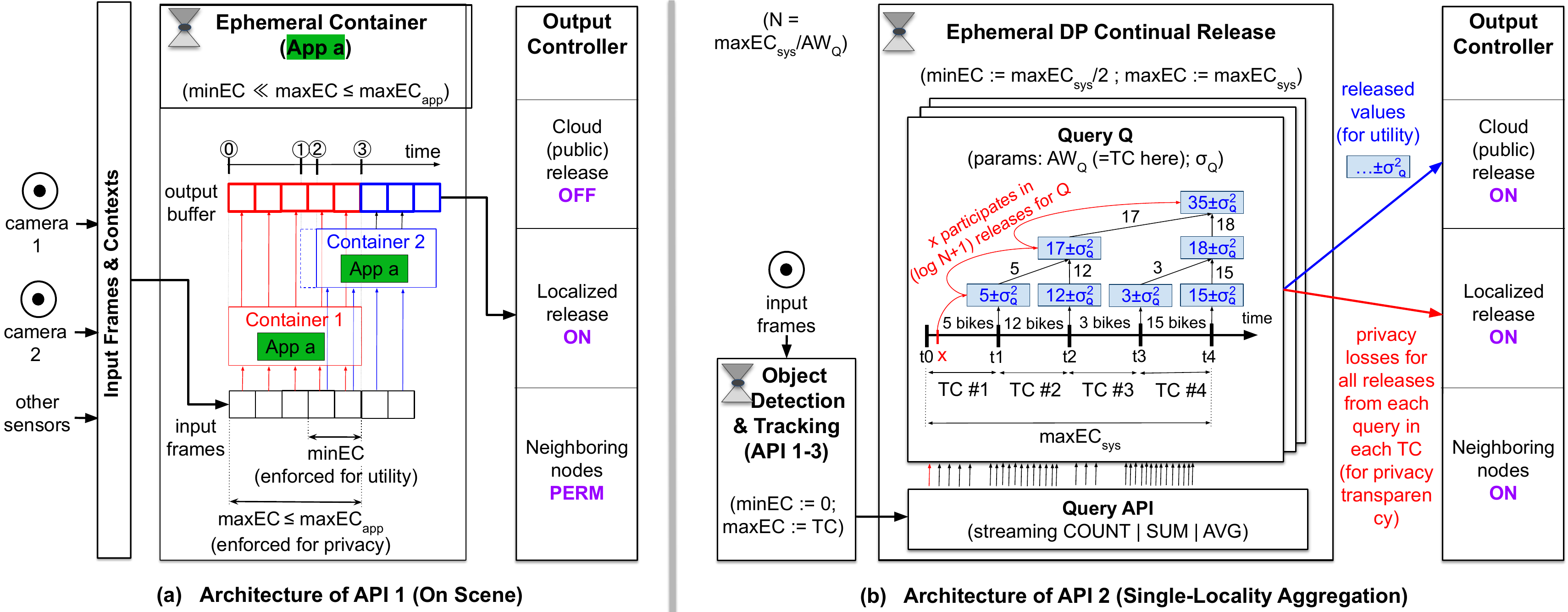}
    \caption{{\bf Detailed design.} (a) API~1 (On-Scene) architecture.
(b) API~2 (Single-Locality Aggregation) architecture.
API~3 (Cross-Locality Aggregation) is illustrated in
\F\ref{fig:api3}. Key mechanisms:
API~1: ephemeral containers, localized release;
API~2: ephemeral DP continual release, bounded tracking contexts;
API~3: relevance self-identification; and cross-cutting:
device-side privacy accounting.}
    \label{fig:detailed-design}
\end{figure*}

\subsection{API 1: On-Scene}
\label{sec:api1}

API~1 supports functions that require access to current sensor data to assess the present state of a space and act on it.
Examples include extending a traffic signal's green phase for a pedestrian
with disabilities, spotting nearby parking, identifying safe delivery stops, and providing autonomous vehicles with an ``eye-in-the-sky'' perspective.

These functions share a common access pattern: the application observes the
current scene, makes a decision, and moves on. Some temporal context is
needed -- enough frames to understand the scene -- but there is no inherent
need to retain data or export it to a backend. Without architectural
constraints, however, developers will be tempted to collect anyway, much as
web and mobile applications routinely repurpose raw-data access for backend
uses such as advertising, recommendations, or model training. API~1 is
designed to preempt that default. \F\ref{fig:detailed-design}(a) shows
the API~1 architecture, which exposes three abstractions:
\emph{input frames and contexts}, \emph{ephemeral containers},
and \emph{localized release}.

\heading{Input frames and contexts.}
A \emph{frame} is the basic unit of sensor input: one or a few consecutive
video frames from a camera, or the corresponding short sampling interval
from another sensor (e.g., audio). A frame may span multiple sensors
captured at approximately the same moment and typically covers much less
than one second.

An \emph{ephemeral context} is a bounded-time unit containing
all input frames in a time window, any non-DP state derived from them,
and any outputs produced from that state. \sysname enforces two parameters
on every ephemeral context: $\maxEC$, which bounds how far into the past outputs may depend on frames, and $\minEC$, which ensures enough temporal context for correct execution. Applications choose $\minEC$ and 
$\maxEC$ values subject to $\minEC \ll \maxEC \le \maxECApp$, where $\maxECApp$ is a system-wide bound for unprivileged applications.
\sysname also defines a larger bound, $\maxECSys$, for privileged
system services such as those used by API~2.
We generally think of $\maxECApp$ on the order of seconds to a
few minutes and $\maxECSys$ on the order of days.

\heading{Ephemeral containers.}
\sysname realizes ephemeral contexts via \emph{ephemeral containers}: 
sandboxed execution environments, implemented as a sequence of Docker
containers, with no network access, persistent storage, or output
except through the Output Controller. The lifetime of the output buffer is
tied to the ephemeral context that produced it.

To the application, execution appears continuous: it observes a
stream of frames and always has access to at least the last
$\minEC$ of context. Internally, \sysname enforces ephemerality
via container rotation. As the active container approaches
$\maxEC-\minEC$, \sysname launches a shadow container, feeds it
the same incoming frames, and switches once it has accumulated
$\minEC$ of state. The old container is then torn down and its
memory destroyed. Thus, applications always receive the required
$\minEC$ context for utility and should be structured to rely
only on it (e.g., operate on a sliding window of size $\minEC$),
while never accessing more than $\maxEC$, for privacy.

From a privacy perspective, therefore, all API~1 outputs depend on at most
the last $\maxEC$ data. This guarantee, however, holds only
within the node: once outputs leave the node, \sysname can no longer
control their lifetime. API~1 therefore couples ephemerality with
localized release.

\heading{Localized release.}
API~1 outputs leave the node only through the Output Controller's
\emph{localized release} channel, the sole channel enabled by default.
Conceptually, this is a bounded-range broadcast: outputs are sent to
devices in the immediate vicinity of the node and are not forwarded to
backend clouds or remote infrastructure. In practice, such scoping could
be implemented over mechanisms such as private 5G~\cite{3gpp-ts23501-r18},
whose coverage area naturally limits dissemination (Appendix~\ref{appendix:api1}).

The Output Controller enforces two additional constraints. First, outputs
inherit the ephemerality of the context that produced them: they are
evicted from the node's output buffer on the same schedule as the
producing container's teardown. Second, forwarding to neighboring
\sysname nodes requires explicit administrative permission and a bounded
hop count. This permits carefully scoped extensions -- e.g., letting a
parking application notify drivers of available spots one intersection
ahead -- without turning localized release into general cross-city
dissemination.

\heading{Sniffing risk and the $\maxEC$ tradeoff.}
Localized release assumes that API~1 outputs largely remain near the node
that produced them, but nearby devices could collect and forward them to
a backend. Shorter $\maxEC$ reduces the amount of history any such
sniffer can recover, at the cost of greater container-rotation overhead
and less temporal context for applications. We evaluate this tradeoff --
privacy vs.\ utility vs.\ overhead -- in \S\ref{sec:eval:api1} and find
that for reasonable parameters, \sysname provides acceptable overhead and
meaningfully limits small-scale sniffing.

Against large-scale sniffing, ephemerality and localization alone do not
suffice. Here \sysname follows the web's emerging response to third-party
tracking: it combines technical and regulatory defenses.
API~1 defaults applications to anonymized inputs -- bounding boxes,
object classes, and tracking IDs from the trusted {\em Object Detection
and Tracking} service -- which reduce, though do not eliminate,
re-identification risk. Beyond that, 
regulation should
constrain systematic collection, retention, and resale of localized
outputs that cannot be fully prevented technically. 

\subsection{API 2: Single-Locality Aggregation}
\label{sec:api2}

Longitudinal statistics at a single location, such as traffic counts over
hours, parking-demand forecasts, and near-colli\-sion frequency, are among
the most common workloads in urban sensing. Without a structured
interface, the path of least resistance is to stream raw video to a
backend and run queries there, accumulating the kind of
persistent behavioral record that \sysname is designed to prevent. API~2
exists to break that default. It provides a streaming aggregation
interface that releases only DP statistics, never raw
observations.

\sysname splits aggregation into two APIs following the same-location
principle (\S\ref{sec:overview}). API~2 handles single-locality
measurements (e.g., per-intersection dashboards, parking occupancy
prediction), implemented as a centralized DP computation at each node.
API~3 handles cross-locality measurements, as discussed in \S\ref{sec:api3}. 

%
\heading{Architecture.}
\F\ref{fig:detailed-design}(b) shows the API~2 architecture. A trusted
system service (\emph{Ephemeral DP Continual Release}) evaluates
streaming queries over a local event stream and periodically emits
differentially private aggregates. Events are produced by the trusted
\emph{Object Detection and Tracking} service running at the node.
Released values are treated as public outputs and routed through the
Output Controller to application backends. The associated privacy loss
is broadcast locally so that nearby devices can track their cumulative
exposure. All non-DP internal state lives inside an ephemeral container
and is discarded upon rotation; the container parameters that make this
work are described below.

\heading{Query API.}
API~2 exposes a streaming query interface over a continuous event stream,
inspired by classic streaming engines~\cite{cetintemel2016aurora}. The
default stream is produced by the trusted \emph{Object Detection} service,
which emits tuples of the form
$\langle \textit{objectID}, \textit{objectType}, \textit{features},
\textit{speed} \rangle$.
Applications register queries using the following template.
\begin{center}\small
\texttt{\{COUNT $\mid$ SUM $\mid$ AVG\}} \texttt{ OVER stream} \\
\texttt{ WHERE } $\textit{objectType} \in T \ \land\ \textit{features} \models \phi$ \\
\texttt{ WINDOWED BY } $\mathit{AW}_Q \ \texttt{ WITH SIGMA } \sigma_Q$
\end{center}

\noindent where $\mathit{AW}_Q$ is the aggregation window and $\sigma_Q$ is
the standard deviation of the noise added upon each release for DP.
To enable efficient multi-scale aggregation, $\mathit{AW}_Q$ is
restricted to dyadic intervals of the form $\maxECSys/2^K$.
\heading{Ephemeral DP continual release.}
API~2's central challenge is reconciling continual
release with data minimization. Standard continual-release
mechanisms~\cite{chan2011private,PUC+25} maintain state across an
unbounded stream so that interval queries can be answered with bounded
noise. \sysname's data-minimization principle
(\S\ref{sec:architecture}), however, requires all non-DP state to be
discarded within $\maxECSys$, even for trusted components. Resolving this tension is API~2's core contribution.

We support two continual-release mechanisms, binary
tree~\cite{chan2011private} and Toeplitz~\cite{PUC+25}, offering
better bounds under different DP regimes (binary tree for pure DP,
Toeplitz for approximate DP). Our evaluation
(\S\ref{sec:evaluation}) uses the latter, but we explain the former
for simplicity. Let $N=\maxECSys/AW_Q$ denote the number of
aggregation windows per ephemeral-container lifetime. In the binary
tree mechanism, each aggregation window corresponds to a leaf in a
complete binary tree of $N$ leaves, and internal nodes represent
unions of contiguous dyadic ranges. At the end of each window, the
system releases the noisy leaf and any newly completed internal
nodes. Each leaf contributes to $\log_2 N+1$ releases, and any query
over a contiguous interval $\ge AW_Q$ can be answered by combining at
most $2(\log_2 N - 1)$ noisy nodes (for $N \geq 4$), yielding
bounded error~\cite{chan2011private}.

\F\ref{fig:detailed-design}(b) illustrates this with the query
{\footnotesize\texttt{COUNT DISTINCT \ldots\ WHERE objectType = bicycle WINDOWED BY 60s WITH SIGMA 1}}.
Assuming each bicycle contributes at most once per 60-second
interval, the system computes the distinct count each minute and
releases the noisy leaf plus newly completed internal nodes. A
dashboard then estimate counts over coarser intervals by summing
released nodes; e.g., the count over $t_1$--$t_3$ uses
$t_1$--$t_2$ and $t_2$--$t_3$, with variance $2\sigma_Q^2$.

To enforce ephemerality while preserving utility, we host the release
service in an ephemeral container with $\maxEC=\maxECSys$ and
$\minEC=\maxECSys/2$. This ensures that any interval of length at
most $\maxECSys/2$ is fully contained within at least one container
instance, so queries over such intervals can be answered from a
single tree with noise comparable to the vanilla binary tree. For
longer intervals, we characterize utility in \S\ref{appendix:api2},
deriving coefficients for the linear growth in error with the number
of ephemeral contexts spanned. Because shadow containers do not
release outputs, the design introduces only one additional noisy node
compared to a tumbling-window baseline ($\log_2 N+2$ versus
$\log_2 N+1$ releases per container lifetime), yielding a
fixed-utility guarantee over sliding windows of size $\maxECSys/2$
with modest additional privacy loss.

\heading{Privacy unit and sensitivity.}
DP requires bounding how much any single individual can influence
released aggregates. On the web, this is straightforward. Sessions,
cookies, and account identifiers give first parties a natural way to
group and cap per-user contributions~\cite{BBD+22}. In public-space
sensing, the infrastructure is passive, observing individuals without
explicit interaction, and thus lacks stable per-person
identifiers. This is a fundamental challenge of the domain, not
specific to \sysname, and means the standard DP-SQL approach
of propagating per-user contribution bounds through query
operators~\cite{BBD+22} cannot be applied directly.

\sysname addresses this by exposing two modes for operators to choose on the formality--utility spectrum.

\heading{Untrusted mode (formal, high noise).}
Following prior work on DP analytics over untrusted video
pipelines~\cite{cangialosi2022privid}, the untrusted mode treats all
events in an aggregation window as potentially originating from a
single individual. This yields a sound DP guarantee without any
identity assumption, but at the cost of high sensitivity. Counting
queries become bounded-sum queries whose noise scales with a declared
count range (e.g., 100 to be able to count up to this many people).
The result is formal but often impractical for fine-grained measurements,
since the noise can overwhelm the signal. Moreover, counting an
unexpectedly large crowd may fail altogether if the true count
exceeds the declared range.

\heading{Trusted mode (empirical, low noise).}
For applications whose queries can be expressed in terms of the
set of objects supported by our trusted \emph{Object Detection and
Tracking} service, \sysname recovers a limited notion of per-person
identity through \emph{bounded tracking contexts} (TCs). A TC is a 
short, system-defined time window (on the order of minutes) within which the trusted service tracks objects using BoT-SORT~\cite{aharon-botsort-2022} as they move through, exit, and
re-enter the scene. Each detected object receives a unique ID that we
use to bound its contribution to queries within that TC.
The service runs in
an ephemeral container with $\maxEC=TC$ and $\minEC=0$; identifiers
are scoped to a single TC, reset across contexts, and never persist
beyond it.

If tracking were perfect, the sensitivity of a count-distinct query
would be exactly~1. In practice, tracking errors
(ID switches due to occlusion, re-entry, or appearance changes) may
assign multiple IDs to a single individual; we empirically
characterize these errors on standard benchmarks and derive a
conservative per-person contribution bound
(\S\ref{sec:eval:api2}). Concretely, we benchmark a production
YOLOv26n + BoT-SORT pipeline on four static-camera MOT17
sequences and find that 56.5\% of individuals receive exactly one
ID, with p90${=}4$, p95${=}5$, p99${=}7$, and a worst-case
maximum of~9. Setting $\Delta{=}9$ covers every individual in the
benchmark and yields an ${\approx}11{\times}$ noise reduction over
the untrusted mode.

The trusted mode's guarantee is therefore \emph{empirical}, resting
on an assumption about tracker fidelity rather than a worst-case
bound. This is analogous to how real-world DP deployments (e.g., DP-SGD for
ML training) routinely rely on empirical assumptions about privacy
units. More importantly, the design surfaces a new problem at the
intersection of computer vision and DP, namely,
developing tracking models whose error bounds can be formally incorporated into
privacy guarantees (e.g., folded into the failure probability of
approximate DP). We leave this to future work.

Under either mode, \sysname enforces \emph{(user, TC)-level} DP at
each node (proofs and parameters in \S\ref{appendix:api2}), maintaining
per-TC budgets that bound total privacy loss per individual.

\heading{Device-side privacy accounting.}
Shorter TCs better align with our limited-tracking principle and allow
more frequent releases, while longer TCs provide a more meaningful
privacy unit. We address this tension by keeping TCs short and enabling
coarser-grained composition on user devices.

At the beginning of each TC, the node assigns it a unique identifier
(tcId) and computes the maximum privacy loss $\rho$ any
individual present during that TC may incur, accounting for all relevant
releases across all queries, including those in
future TCs. The node then continuously broadcasts
$\langle \text{tcId}, \rho \rangle$ locally and we assum that if
an individual can be sensed, the individual also hears the broadcast.
Devices record this loss once per TC. Thus, although API~2 enforces
privacy separately at each node and over short TCs, devices
compose losses into a longitudinal, city-wide view of privacy cost.

\subsection{API 3: Cross-Locality Aggregation}
\label{sec:api3}

\begin{figure}[t]
    \centering
    \includegraphics[width=\linewidth]{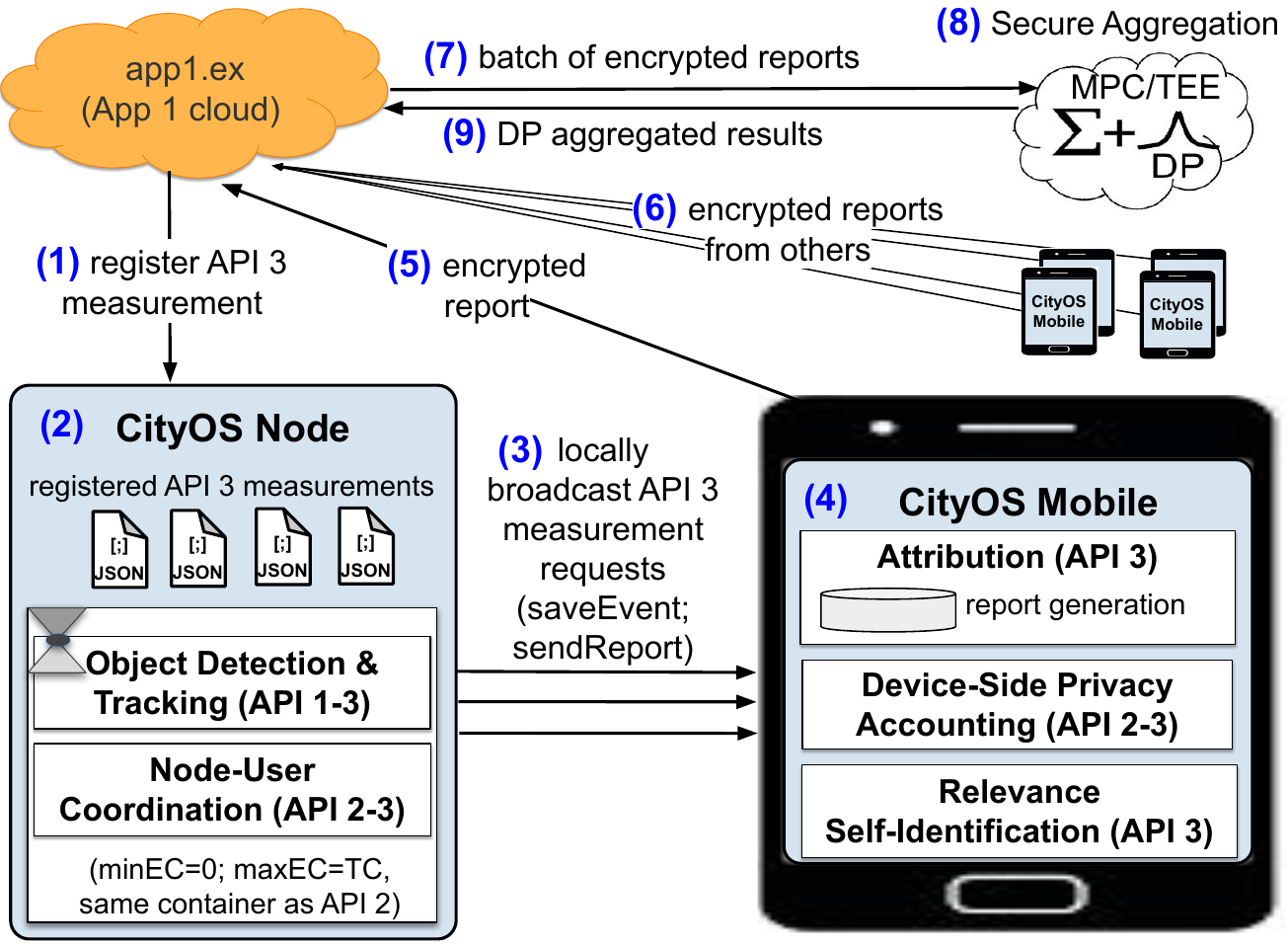}
    \caption{{\bf API 3: Cross-Locality Aggregation.} Key concepts: device-mediated federated measurement, on-device IDP enforcement, and node-user coordination for context-aware participation.}
    \label{fig:api3}
\end{figure}

Cross-location analytics, such as pedestrian flows between neighborhoods, transit origin--destination estimation, and how infrastructure changes affect foot traffic, rely on centralized pipelines that reconstruct individual trajectories from mobile-location data~\cite{Cuebiq-billboard-conversion-measurement}. API~3 replaces these pipelines with a federated interface adapted from W3C Attribution, where user devices produce encrypted, privacy-budgeted reports and infrastructure never observes individual trajectories.

\heading{Use cases.}
API~3 supports measurements over spaces that span multiple \sysname nodes and have well-defined entry and exit boundaries, such as commercial corridors, transit systems, or highway segments. While users remain within such a region, their devices locally record limited state (e.g., entry time or coarse trajectory features); upon exit, they generate encrypted reports that are aggregated with DP.

This model supports estimating pedestrian flows between neighborhoods, measuring transit-route popularity, and evaluating the impact of infrastructure changes without persistent identifiers or infrastructure-based tracking.

\heading{Architecture.}
\F\ref{fig:api3} shows the API~3 architecture, which splits responsibility between infrastructure and user devices.
Applications register measurement tasks with \sysname nodes in the target region, specifying simple device-side actions such as recording entry time or reporting elapsed time upon exit. Nodes store these tasks and broadcast them locally to nearby devices. Devices that receive a task execute it locally, recording minimal state while the user is in the region and, upon exit, generating encrypted reports that encode bounded statistics (e.g., time spent, or coarse trajectory summaries). These reports are sent to the application backend, which aggregates them across users using secure aggregation (e.g., MPC or TEE) and releases only DP results. Throughout, infrastructure never observes individual trajectories; nodes determine \emph{when} and \emph{which} measurements to perform, but all sensitive state and privacy enforcement remain on devices. This workflow mirrors Attribution's own primitives adapted to physical space; \S\ref{appendix:api3} gives an intuitive mapping, based on which we claim \sysname's API~3 DP guarantees; a formal mapping and analysis are pending work at the time of this writing.

\heading{Privacy model.}
Unlike API~2, where privacy is enforced at the node, API~3 enforces and accounts for privacy entirely on-device. Each device maintains an individual DP (IDP) budget over coarse-grained epochs. Each report consumes budget; once exhausted, further contributions are suppressed. \sysname reuses Attribution's on-device IDP accounting and composes it with API~2 losses, exposing a unified privacy-accounting interface. API~3 therefore provides strong per-device, per-epoch guarantees without requiring infrastructure to maintain per-user state.

%
\heading{Subpopulation disambiguation.}
The design above assumes measurements over all users in a region. In practice, applications often require subpopulation-specific statistics, such as pedestrians vs.\ cyclists or trucks vs.\ passenger vehicles. On the web, Attribution supports such distinctions because the querier interacts directly with the user (e.g., via an HTTPS session at conversion time) and can assign reports to measurement classes. In \sysname, nodes have no such direct interaction with devices; absent additional mechanisms, devices would participate based only on location, not on the querier's intended population. This is a utility challenge rather than a privacy one, and it requires a coordination mechanism that does not itself create new tracking vectors.

This raises a core systems question, namely, how to coordinate infrastructure and user devices to recover this capability without introducing new privacy-leakage channels.

\heading{Node-user coordination.}
\sysname addresses this through a two-sided mechanism that combines local broadcast with on-device inference. On the node side, the Node-User Coordination service runs alongside the Object Detection and Tracking service within a bounded tracking context. For each detected object, the node attaches metadata indicating which API~3 measurements are relevant (e.g., a bicycle-only measurement is attached only to bicycle detections; a corridor-entry measurement only to objects observed entering the region). These per-object annotations are broadcast locally using the same localized-release channel as API~1.

On the device side, upon receiving annotated signals, the device determines whether any of the broadcast tracks correspond to the user carrying it. It does so using onboard IMU data to classify coarse activity (pedestrian, cyclist, driver) and, when finer disambiguation is needed, to match its motion trace against nearby object tracks via cadence or motion-pattern analysis. The device then participates only in measurements for which it identifies itself as relevant. We call this process {\bf relevance self-identification} and it is performed by an on-device homonymous component (\F\ref{fig:architecture}). In simple cases, activity classification alone suffices; in more complex ones (e.g., entry vs.\ exit, or truck vs.\ car), track matching leverages infrastructure as a source of disambiguating context that mobile sensing alone cannot provide, while all identification logic remains on-device.          
\section{Prototype and Applications}
\label{sec:apps}

We implement \sysname in ${\sim}$4,100~LoC of Go, ${\sim}$530~LoC of Python, and ${\sim}$2,300~LoC of Rust, and build \numAppsUntrusted applications, plus the Object Detection and Tracking service, itself an API~1 app. We describe the \numAppsUntrusted applications below, two of which are illustrated via workflow visualizations in \F\ref{fig:app-visualizations}.

\begin{figure*}[t]
    \centering
    \includegraphics[width=\linewidth]{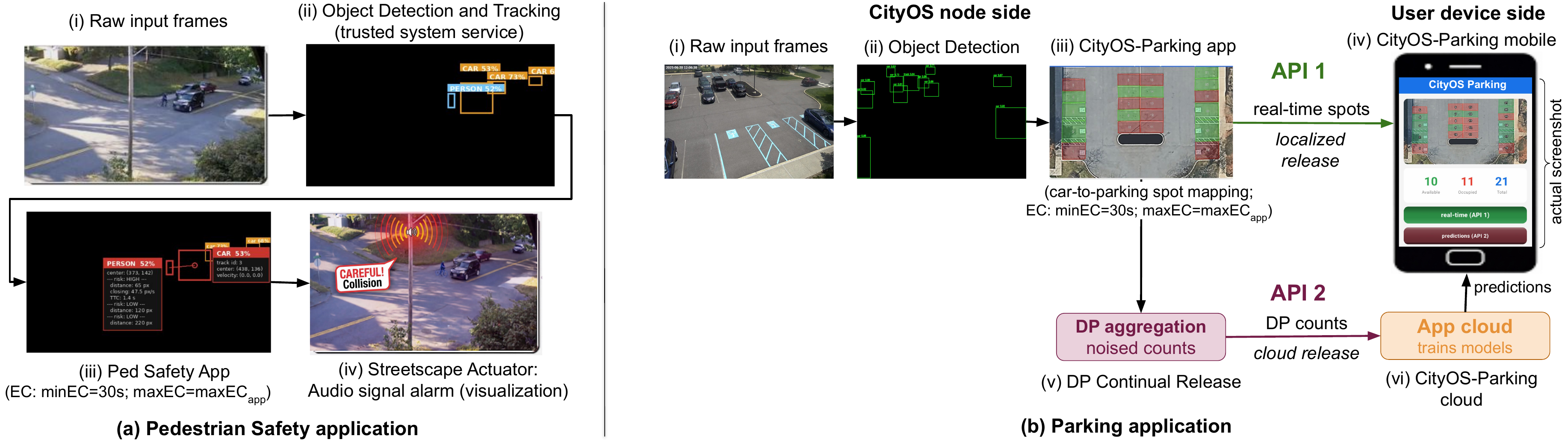}
    \caption{{\bf \sysname application workflows.} (a) Pedestrian safety (API~1). (b) Parking app (API~1+API 2). 
    }
    \label{fig:app-visualizations}
\end{figure*}

\heading{Pedestrian Safety (API~1).}
(\F\ref{fig:app-visualizations}a)
Provides real-time alerts for imminent pedestrian-vehicle collisions.
The workflow is: (i) raw input frames are processed by (ii) the trusted
Object Detection and Tracking service, which produces pedestrian,
cyclist, and vehicle tracks; (iii) the pedestrian safety application,
running in an API~1 ephemeral container, uses these tracks to estimate
collision risk based on distance and relative velocity over time.
Collision detection requires temporal context, as distinguishing a true
collision course from momentary proximity requires several seconds of
motion history. The application thus requests $\minEC=2$s, with
$\maxEC \le \maxECApp$ enforced by the system. When risk exceeds a
threshold, the application (iv) emits alerts via localized release to
nearby streetscape devices (e.g., speakers).

\heading{Parking App (API~1+API~2).}
(\F\ref{fig:app-visualizations}b)
Spans two API tiers: (1) to provide real-time spot availability to
nearby users via a companion mobile app (API~1); and (2) to provide
future availability predictions by training on DP aggregates of parking
occupancy (API~2). The API~1 workflow is:
(i) raw input frames are processed by (ii) the trusted Object Detection
and Tracking service, which outputs bounding boxes with object classes;
(iii) the parking app maps detected vehicles to predefined parking-spot
polygons (shown in the figure using an aerial view of the lot),
identifying occupied and available spots based on position.
The result is a per-frame list of spot IDs marked available or occupied.
This information is broadcast locally (and, with administrative
permission, within a small number of hops), enabling (iv) nearby users
running the companion mobile app to receive real-time availability.

In parallel, the API~2 workflow is: (v) the app registers a counting query with the DP Continual Release service to obtain DP estimates of parking occupancy over an aggregation window (e.g., one hour). (vi) These DP releases are sent to the application backend, which trains a simple predictive model (e.g., based on time of day and day of week). The mobile app can then request availability forecasts while the user is still far from the destination.

\heading{Dashboard App (API~2).}
(\F\ref{fig:apps:dashboard})
Provides per-location traffic statistics for city planners.
The application registers counting queries for pedestrians, vehicles, and
bicycles with the API~2 DP Continual Release service (\S\ref{sec:api2}),
which aggregates events over time and forwards differentially private
counts to the backend, where the dashboard visualizes them as hourly
breakdowns with confidence intervals.

\heading{Subway Trajectories (API~3).}
Estimates the popularity of subway routes without tracking riders.
The NYC subway lacks accurate route-level statistics because riders
tap in but do not tap out, forcing reliance on indirect, per-station
estimates. The workflow is: (i) \sysname nodes deployed at turnstiles
(entry/exit points) broadcast measurement requests to nearby devices;
(ii) devices record entry location and time upon entering the system;
(iii) upon exit, devices submit encrypted, privacy-budgeted reports
encoding entry location and elapsed time; and (iv) the backend aggregates
these reports with DP to estimate route popularity
and travel times. Devices use node-user coordination to self-identify
as entrants or exiters and execute the appropriate actions. This enables 
route-level measurement without persistent identifiers, tracking, or face
recognition, relying instead on device-mediated, privacy-bounded reporting.

\section{Evaluation}
\label{sec:evaluation}

We evaluate \sysname primarily at the API-mechanism level, using a representative workload per tier to measure the privacy-utility tradeoffs induced by each abstraction. These workloads exercise the core mechanisms that underlie the prototype applications in \S\ref{sec:apps}: ephemeral localized execution (API~1), DP continual release (API~2), and device-mediated cross-locality aggregation with self-identification (API~3). 
We evaluate: how effectively $\maxEC$ limits sniffing attacks and at what latency cost (API~1), whether DP continual release preserves useful counting accuracy (API~2), and whether cross-locality aggregation recovers useful statistics under per-device privacy budgets, and how self-identification quality affects that utility (API~3).

All node-side experiments run on an Nvidia Jetson AGX Orin edge GPU,
representative of embedded edge accelerators; we also benchmark API~1 on a desktop 5060~Ti (\S\ref{appendix:eval-details:api1}). Our primary accuracy
metric is RMSRE (root mean square relative error~\cite{hay-boosting-2010}), the RMS of per-bin relative deviations from ground truth. We use Toeplitz for DP continual release. Methodology details (hardware configuration, dataset preprocessing, noise parameters) are in \S\ref{appendix:evaluation-details}.


\subsection{API~1 On-Scene}
\label{sec:eval:api1}

API~1's main knob is $\maxEC$, the ephemeral context's lifetime. Shorter ECs reduce what a nearby sniffer can recover from any one location, but they also increase container-rotation overhead. We therefore evaluate both sides of this tradeoff, namely privacy exposure under sniffing attacks from \S\ref{sec:api1} and latency overhead in Object Detection and Tracking services.

\heading{Sniffing Attack with $\maxEC$.}
We simulate a small-scale sniffing adversary traversing Midtown Manhattan and, upon reaching an intersection, reads the API~1 output buffer and saves it. We maliciously grant the attacker the full buffer contents at arrival. We vary $\maxEC$ across \{30\,s, 1\,min, 2\,min, 3\,min, 4\,min, 5\,min\} and report the fraction of total city-wide person-activity captured over one hour. We consider four attacker profiles that traverse the evaluated perimeter at different speeds, namely a pedestrian, a \emph{cyclist}, a \emph{car}, and a {\em static device} planted at Times Square (speeds in \F\ref{fig:api1-sniffing}).

\begin{figure*}[t]
\centering

\begin{subfigure}[t]{0.26\linewidth}
    \vspace{0pt}
    \centering
    \includegraphics[width=\linewidth]{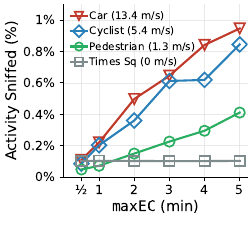}\vspace{-0.2cm}
    \caption{API~1 sniffing attack}
    \label{fig:api1-sniffing}
\end{subfigure}
\hfill
\begin{subfigure}[t]{0.19\linewidth}
    \vspace{0pt}
    \centering
    \scriptsize
    \begin{tabular}{@{}lr@{}}
    \toprule
    \multicolumn{2}{c}{\textbf{\# track IDs / person}} \\
    \midrule
    Mean   & 1.92 \\
    Median & 1    \\
    p90    & 4    \\
    p95    & 5    \\
    p99    & 7    \\
    Max    & 9    \\
    \midrule
    ${=}1$ ID        & 56.5\% \\
    ${\leq}2$ IDs    & 77.0\% \\
    ${\leq}3$ IDs    & 88.5\% \\
    ${\leq}5$ IDs    & 96.9\% \\
    \bottomrule
    \end{tabular}\vspace{0.2cm}
    \caption{API~2 sensitivity study}
    \label{fig:api2-sensitivity}
\end{subfigure}
\hfill
\begin{subfigure}[t]{0.26\linewidth}
    \vspace{0pt}
    \centering
    \includegraphics[width=\linewidth]{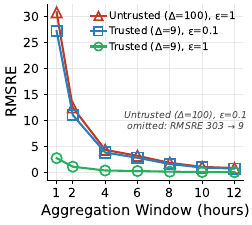}\vspace{-0.2cm}
    \caption{API~2 bike counter accuracy}
    \label{fig:api2-counter}
\end{subfigure}
\hfill
\begin{subfigure}[t]{0.26\linewidth}
    \vspace{0pt}
    \centering
    \includegraphics[width=\linewidth]{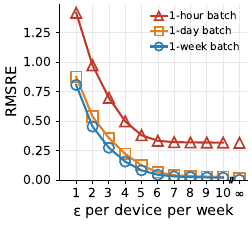}\vspace{-0.2cm}
    \caption{API~3 subway OD accuracy}
    \label{fig:api3-subway}
\end{subfigure}

\vspace{1em}
\caption{{\bf Evaluation of \sysname APIs.}
(a) API~1 sniffing-attack effectiveness under time-scoped, localized outputs.
(b) API~2 empirical sensitivity study.
(c) API~2 DP counting accuracy (trusted vs.\ untrusted).
(d) API~3 subway origin-destination (OD) accuracy vs.\ per-device budget.}
\label{fig:apis-eval}
\end{figure*}

\F\ref{fig:api1-sniffing} shows that $\maxEC$ is an effective privacy knob against mobile sniffers. For the fastest attacker (car), capture rises from $0.11\%$ of city-wide person-activity at $\maxEC{=}30$\,s to $0.95\%$ at $5$\,min (${\sim}9\times$); pedestrians and cyclists follow the same trend, rising from $0.05\%$ to $0.41\%$ and from $0.09\%$ to $0.84\%$, respectively; the car visits roughly 147 of the ${\sim}540$ intersections (27\%) yet captures under 1\% of city-wide activity. A static attacker captures ${\sim}0.1\%$ regardless of $\maxEC$, since it observes only one location. Together, these results show that ephemerality primarily limits what \emph{mobile} sniffers can accumulate across locations, while localization confines any \emph{static} compromise to its immediate vicinity.


\heading{Performance overhead.}
On a desktop 5060~Ti, frame drops reach near-zero by $\maxEC{\geq}10$\,s; the Jetson AGX Orin plateaus near ${\sim}17\%$ drops due to slower container startup (\F\ref{fig:api1-latency}; details in \S\ref{appendix:eval-details:api1}).


\subsection{API~2 Single-Locality Aggregation}
\label{sec:eval:api2}

API~2 replaces raw longitudinal collection at a node with DP continual
release. The key question is whether releases remain useful. We
compare \emph{trusted} mode against the \emph{untrusted}
mode. The trusted mode has empirical but tight sensitivity; the untrusted mode provides sound but high sensitivity.

\heading{Empirical sensitivity.}
Trusted-mode sensitivity $\Delta$ equals the maximum track-ID multiplicity a single person produces per release period ($\rho_{\mathsf{track}}$ in our formal model, \S\ref{appendix:api2}). We benchmark the production YOLOv26n + BoT-SORT pipeline on four static-camera MOT17 sequences~\cite{milan-mot16-2016} (3{,}012 frames, 191 matched identities; pipeline and matching details in \S\ref{appendix:eval-details:api2}). \F\ref{fig:api2-sensitivity} summarizes the distribution. The median is 1~ID per person (56.5\% receive exactly one), with p90${=}4$, p95${=}5$, p99${=}7$, and a worst-case maximum of~9. We conservatively set $\Delta{=}9$ so the DP guarantee holds for every individual in the data. Using p95 ($\Delta{=}5$) reduce noise $1.8{\times}$ while still covering 96.9\% of individuals; p99 ($\Delta{=}7$) covers 99\% with $1.3{\times}$ less noise. For the untrusted path, $\Delta{=}100$ (operator bound), yielding an ${\approx}11{\times}$ sensitivity gap that propagates linearly into noise.

\heading{Counting accuracy.}
We evaluate on 30 days of 15-minute bicycle counts from the Williamsburg Bridge Eco-Counter~\cite{nyc-dot-bike-counters}, aggregating over windows from 1 to 12~hours and reporting average RMSRE over 1{,}000 Monte Carlo trials.

\F\ref{fig:api2-counter} shows the results across three configurations: trusted at $\varepsilon{=}1$, trusted at $\varepsilon{=}0.1$ (a $10{\times}$ stricter budget), and untrusted at $\varepsilon{=}1$. The untrusted mode incurs a fixed ${\approx}11{\times}$ penalty across all release periods, reflecting the linear scaling of noise with sensitivity. Aggregation dominates accuracy: at $\varepsilon{=}1$, trusted RMSRE drops from $2.77$ (1~hour) to $0.29$ (6~hours) to $0.08$ (12~hours), a ${\sim}35{\times}$ improvement at fixed privacy cost. At p95 ($\Delta{=}5$) these values would be $1.8{\times}$ lower. Even with a $10{\times}$ smaller $\varepsilon{=}0.1$, the trusted path (RMSRE~$2.84$ at 6~hours) still beats untrusted at $\varepsilon{=}1$ (RMSRE~$3.21$), because trusted mode reduces sensitivity by ${\approx}11{\times}$, more than compensating for the tighter budget.

\begin{figure*}[t]
\centering
\begin{subfigure}[t]{0.24\linewidth}
    \vspace{0pt}
    \centering
    \includegraphics[width=\linewidth]{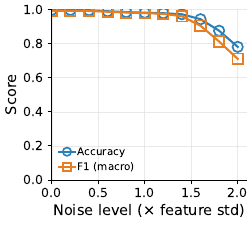}
    \caption{Phase~1 IMU activity classification}
    \label{fig:self-id-noise}
\end{subfigure}
\hfill
\begin{subfigure}[t]{0.24\linewidth}
    \vspace{0pt}
    \centering
    \includegraphics[width=\linewidth]{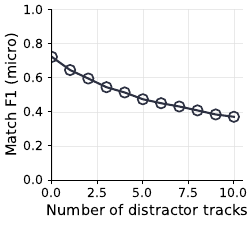}
    \caption{Phase~2 matching with distractors}
    \label{fig:stage2-distractors}
\end{subfigure}
\hfill
\begin{subfigure}[t]{0.24\linewidth}
    \vspace{0pt}
    \centering
    \includegraphics[width=\linewidth]{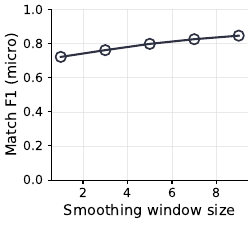}
    \caption{Phase~2 matching with smoothing}
    \label{fig:stage2-smoothing}
\end{subfigure}
\hfill
\begin{subfigure}[t]{0.24\linewidth}
    \vspace{0pt}
    \centering
    \includegraphics[width=\linewidth]{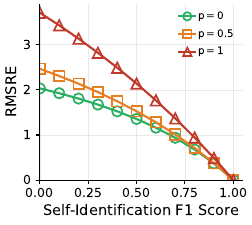}
    \caption{Subway OD utility}
    \label{fig:api3-self-id}
\end{subfigure}
\vspace{1em}
\caption{{\bf API~3 self-identification evaluation.} (a)--(c) Accuracy of our two-stage self-identification pipeline. (d) Impact of identification errors on subway OD.}
\label{fig:self-identification}
\end{figure*}

\subsection{API~3 Cross-Location Aggregation}
\label{sec:eval:api3}

API~3 is the strongest test of the architecture, aiming to recover cross-location statistics without allowing infrastructure to track individuals. We evaluate two questions. First, how much utility is lost when devices enforce per-user privacy budgets and stop reporting after exhaustion? Second, how sensitive is that utility to imperfect self-identification?
We use the subway origin-destination (OD) estimation application from
\S\ref{sec:apps}. We answer the first question here and the second in the subsequent section.

We use the Hangzhou Metro smart-card dataset (29.2\,M trips by
4.97\,M riders across 80 stations over 25 days). Each entry event is
stored on-device, and each exit generates a report charging
$\varepsilon_{\text{rep}}{=}0.5$ against a weekly device budget
$\varepsilon$. Reports are DP-aggregated per exit station to produce an OD histogram. We vary $\varepsilon$ across
$\{1,\dots,10,\infty\}$ and batch lengths of one hour, one day, and one week, reporting RMSRE.

\F\ref{fig:api3-subway} shows three regimes. First, weekly batches isolate the effect of budget exhaustion. Because bins are large (${\sim}1{,}140$ trips/bin), the DP noise floor is tiny (RMSRE $0.0025$), so error is driven mainly by how many riders exhaust their budgets. RMSRE drops sharply from $0.81$ at $\varepsilon{=}1$ to $0.083$ at $\varepsilon{=}5$ and to $0.020$ at $\varepsilon{=}10$. This reflects the rider distribution, since 93\% of riders make at most 10 trips per week and are therefore fully covered once $\varepsilon \ge 5$.
Second, daily batches trade some accuracy for finer time resolution. RMSRE falls from $0.88$ at $\varepsilon{=}1$ to $0.13$ at $\varepsilon{=}5$, $0.042$ at $\varepsilon{=}7$, and $0.028$ at $\varepsilon{=}10$, approaching the daily noise floor of $0.016$, accurate enough for ranking high-volume corridors without constructing infrastructure-side trajectories. Hourly batches are noise-limited (RMSRE $0.32$ even at $\varepsilon{=}\infty$ with ${\sim}9$ trips/bin), placing the operating point at daily-to-weekly aggregation.


\subsection{API~3 Self-Identification}
\label{sec:eval:self-idenfification}

The preceding section assumed that devices can determine whether they are relevant to a measurement (e.g., whether they are entering vs. exiting a subway station). We now evaluate this assumption, studying (1) the accuracy of our two-stage self-identification pipeline and (2) how identification errors affect cross-locality utility. Non-vital methodological details for each experiment are in \S\ref{appendix:eval-details:self-idenfification}.


\heading{Stage~1 activity classification from IMU.} 
We classify device activity (walking, driving, cycling) from IMU signals using a gradient-boosting model trained on standard IMU datasets spanning pedestrian, vehicle, and cycling traces~\cite{vi-fi_2022,driver-2022,pamap2}. To measure robustness, we inject additive Gaussian noise scaled by the per-feature standard deviation of the test set (x-axis, so 0.0 is clean data, 1.0 adds noise with the same spread as the original features, and 2.0 doubles it), retraining at each scale. \F\ref{fig:self-id-noise} shows that Stage~1 is highly accurate and robust to noise. On clean data, it achieves 99.7\% accuracy and 99.2\% macro F1, with per-class F1 of 99.9\% (Walking), 99.8\% (Driving), and 97.9\% (Cycling). With noise-augmented training, it retains 98.1\% accuracy (97.8\% F1) at noise equal to one feature standard deviation, degrading gracefully to 78.0\% accuracy at 2$\times$ noise. This places the system in a high-F1 regime under realistic sensing conditions.

\heading{Stage~2 bounding-box to IMU matching.}
Stage~2 matches pedestrian phones to visual tracks by comparing walking-band (0.7--3.0\,Hz) power spectral densities of each phone's IMU signal and each track's visual motion, forming a pairwise cosine-distance matrix, and solving the assignment with the Hungarian algorithm. An adaptive null-assignment penalty routes unmatched tracks to a dummy, so extra visual tracks do not force incorrect pairings. We evaluate using window-level micro-averaged F1 on outdoor pedestrian traces from Vi-Fi~\cite{vi-fi_2022}.
With no distractors, Stage~2 achieves F1 of $0.72$. \F\ref{fig:stage2-distractors} shows that injecting up to 10 distractor tracks degrades F1 gradually. \F\ref{fig:stage2-smoothing} shows that temporal smoothing over 5--7 windows further improves stability by correcting transient mismatches. This demonstrates that cross-modal matching remains reliable in crowded scenes.


\heading{Impact on API~3 utility.}
We evaluate how self-identification quality affects the subway OD measurement from \S\ref{sec:eval:self-idenfification}, which assumed perfect classification. In practice, transfer passengers may be misclassified as exiters, injecting spurious reports. We reuse the Hangzhou Metro dataset and infer 13.6\,M transfer events. We model self-identification with recall $a \in [0,1]$ and false positive rate $1{-}a$, vary station geometry via $p \in \{0, 0.5, 1\}$ (the fraction of transfers near exits), and report weekly RMSRE including the analytic noise floor ($0.0025$).

\F\ref{fig:api3-self-id} shows that classification quality directly determines OD utility. At F1${=}1$, all curves collapse to the noise floor; at F1${=}0.5$, RMSRE reaches $2.14$ in the worst case ($p{=}1$) versus $1.35$ with no transfer traffic ($p{=}0$), showing that false exits dominate the error budget. Missed exits reduce counts in bins that already carry signal, but false exits both inflate counts \emph{and} introduce nonzero entries in zero-valued bins, producing disproportionately large relative errors. Station layouts that separate transfer paths from exits (low $p$) suppress the more damaging error mode. Degradation is gradual. At F1${=}0.9$, worst-case RMSRE is $0.49$, sufficient for corridor-level rankings. Since Stage~1 achieves $97.8\%$ F1 under realistic noise, the system operates in a regime where RMSRE approaches the noise floor, showing \sysname supports accurate OD estimation without infrastructure-level tracking.

\section{Related Work }
\label{sec:related-work}

\sysname draws on four lines of prior work, each exposing a facet of
the same gap, namely the absence of privacy-enforcing interfaces at
the data-access layer. Inspired by the web's shift toward
privacy-conscious platform
APIs, \sysname fills this gap with purpose-built, risk-calibrated
interfaces.

\heading{Urban sensing.}
Urban sensing spans fixed networks~\cite{murty-citysense-2008}, mobile
collection~\cite{hull-cartel-2006}, and participatory
paradigms~\cite{burke-participatory-2006}. Multi-year
deployments~\cite{catlett-aot-2017,sanchez-smartsantander-2014,mydlarz-sonyc-2019}
show that privacy concerns surface only after hardware is fielded;
\sysname defines data-access constraints before applications are
deployed.

\heading{Platforms and interoperability.}
sMAP~\cite{dawson-haggerty-smap-2010} showed that standardizing
interfaces unifies heterogeneous feeds. City-scale
standards~\cite{etsi-ngsild-2024,ogc-sensorthings-2016,onem2m-ts0001-2020}
and platforms built on
them~\cite{jeong-citydatahub-2020,fiware-2016,iudx-2020} centralize
raw streams before evaluating access policy. \sysname intervenes
earlier so re-identifiable data never reaches a backend store.

\heading{Edge and fog computing.}
NIST~\cite{nist-fog-2018} and ETSI MEC~\cite{etsi-mec-2024} formalize
compute near endpoints, driven by bandwidth and latency constraints on
modalities like video~\cite{rivas-cova-2022}.
Waggle/SAGE~\cite{beckman-waggle-2016} embodies ``process locally,
export derived products'' but provides no mechanism to bound what an
edge plugin may export; \sysname adds this missing layer through its
three API tiers.

\heading{Privacy and governance.}
Anonymization techniques do not compose under
linkage~\cite{dwork-roth-2014}. Legal
frameworks~\cite{gdpr-2016,nist-privacy-2020,nistir-8259a-2020,etsi-en303645-2020}
specify obligations but not enforcement; federated
learning~\cite{pandya-fl-smartcities-2023} protects training data but
not inference flows. The closest precedent, W3C Attribution
API~\cite{w3c-attribution-2024,TKM+24}, replaced cross-site tracking
with differentially private measurement; \sysname applies this
principle to physical sensing.
\section{Conclusion}
\label{sec:conclusions}

We presented \sysname, a privacy architecture for urban sensing
grounded in a web-to-physical-world data-access analogy, with three
API tiers enforced by ephemeral containers, localized release, DP
continual release, and node-device coordination. Five prototype
applications demonstrate practical smart-space functionality with
privacy guarantees.

\section*{Acknowledgments}
\label{sec:acks}

This work was supported by the National Science Foundation (NSF) and Center for Smart Streetscapes (CS3) under NSF Cooperative Agreement No. EEC-2133516.

We thank CS3 members for discussions and feedback about the project, especially Jason Hallstrom, Henning Schultzrinne, and Andrew Smyth.

A number of students contributed to the development of \sysname: Kimberly Collins (container lifecycle fixes), Othmane El~Houssi (frame garbage collection and output driver acknowledgements), Ian Kammerman (integration testing and macOS support), Christian Montgomery (stream access control), Apolline Weinstein (inter-node communication and output buffer API), and Seojin Yoon (video and webcam input drivers), all at Columbia University; and Maahin Rathinagiriswaran (IMU classifier) at Rutgers University. Pierre Tholoniat, Jessica Card, and Alison Caulfield (Columbia University) contributed to the development of the Attribution library we are using, including core DP accounting, Python bindings, and logging infrastructure, respectively. We thank them all for their valuable work.

This article was edited with the assistance of generative AI. All ideas, design details, analyses, and results originate from the authors. The text was initially written by the authors, then edited with AI for clarity, concision, and structure, and finally verified by the authors for accuracy and authenticity. The \sysname code was developed in a similar author-originated, AI-assisted, and author-validated manner.
                                    
\clearpage
\bibliographystyle{ACM-Reference-Format}
\bibliography{refs/all,refs/bias_pierre,refs/jorge}

\clearpage
\newpage
\appendix

\section{API~1: On-Scene -- Security Analysis}
\label{appendix:api1}

API~1 exposes the richest access to raw sensing data and therefore
carries the highest immediate risk if not carefully constrained.
We summarize its protections and residual risks under our threat model
(\S\ref{sec:threat-model}). This is a security analysis, not a privacy
one, as API~1 provides no formal privacy guarantees.

\heading{Node compromise.}
If an adversary compromises a \sysname node, the exposed data is
fundamentally limited by API~1's ephemerality guarantees. At any time,
the node retains at most $\maxECSys$ worth of raw input frames and
derived state. As a result, temporary compromise yields only a short,
temporally bounded snapshot of activity at that location, rather than
a persistent or historical record. This sharply contrasts with
centralized architectures, where compromise exposes long-term stored
data across many locations. Long-term compromise (such as with advanced persistent threats) is outside the scope of all \sysname defenses.

\heading{Application-level exfiltration.}
Untrusted applications may attempt to exfiltrate information through
API~1 outputs. \sysname constrains this channel in two ways.
First, applications execute inside ephemeral containers with no network
access or persistent storage, preventing direct backend transmission.
Second, outputs are restricted to \emph{localized release}, a bounded-range
broadcast channel that does not forward data to remote infrastructure.

This design prevents straightforward backend collection but leaves open
the possibility of indirect exfiltration via nearby devices.

\heading{Localized release via private 5G.}
Localized release can be implemented with a private 5G small cell
co-located with each \sysname node~\cite{peterson-private5g-2023}.
Recent spectrum-sharing regulations allow organizations to deploy
private cellular networks without acquiring traditional licensed
spectrum. A small cell mounted on a streetlight covers roughly
100--300\,m in urban settings~\cite{rodriguez-pathloss-2013}, so radio propagation itself bounds
dissemination to the node's immediate vicinity without application-level
geo-fencing. The 5G architecture further supports deploying the network's
data-forwarding function at each site~\cite{3gpp-ts23501-r18,peterson-private5g-2023,macdavid-upf-2021}, so that application traffic is
served entirely on the node and never traverses a wide-area backhaul
link, so outputs remain local by construction.
For mobile receivers, cellular handover completes in
tens of milliseconds~\cite{hassan-vivisecting-2022}, giving even a vehicle
at city speeds ample dwell time within a single cell. The coverage boundary is a practical, not
cryptographic, spatial bound; we rely on ephemerality ($\maxEC$) rather
than on radio propagation alone for privacy (\S\ref{sec:api1}).

\heading{Sniffing attacks.}
A nearby device can collect API~1 outputs and forward them to a backend
(a \emph{sniffing attack}, \S\ref{sec:api1}). The effectiveness of such
attacks depends on their scale.

\emph{Small-scale attackers} (few devices) are strongly limited by
ephemerality and locality. As shown in \S\ref{sec:eval:api1},
even a car-speed attacker captures at most $2.22\%$ of city-wide
person-activity in one hour under the largest $\maxEC$ tested
(10\,min), and as little as $0.11\%$ under short contexts (30\,s).
The attacker obtains only fragmentary, temporally bounded observations
from visited locations, not a continuous city-wide view.

\emph{Static attackers} capture data only at a single location, with
exposure determined by local traffic rather than system-wide scale
(\S\ref{sec:eval:api1}). Localized release confines this risk
spatially: compromise does not propagate across nodes.

\emph{Large-scale attackers} (e.g., widely deployed mobile apps) could
collect outputs across many locations. While this case is outside our
core threat model, \sysname mitigates it through (i) anonymized inputs
by default (bounding boxes rather than raw imagery), reducing the
information content of outputs; and (ii) reliance on policy and
regulation to prohibit systematic collection and resale of localized
signals, analogous to web defenses against third-party tracking.

\heading{Targeted exfiltration.}
Applications may attempt to encode information in seemingly benign
outputs, such as watermarking signals, tied to specific targeted individuals.
Such attacks are difficult to detect purely technically. \sysname limits
their impact through ephemerality (bounded temporal scope) and
localization (bounded spatial scope), ensuring that successful
attacks require either a large-scale sniffer or physically trailing
the individual.

\heading{Visual feedback loops.}
A potential attack against \sysname's cross-ephemeral-context
independence is a \emph{visual feedback loop} between a malicious
application at the node and a colluding device in the environment.
In this attack, when the application detects a target individual,
it emits a signal via API~1 outputs; a nearby device senses this
signal and responds with a visual (or otherwise observable) pattern
that is fed back into the node's sensors. By iterating this loop,
the application can effectively externalize state and carry it across
ephemeral contexts.

Such behavior is feasible and should be explicitly disallowed at the
policy level. Its impact depends on the attacker's objective.
If the goal is \emph{information exfiltration}, the attacker gains
no additional power beyond a sniffing attack: any information encoded
in the feedback loop could equivalently be captured directly from
API~1 outputs by the nearby colluding device (\S\ref{sec:api1}).
Thus, ephemerality and localization continue to bound the attacker's
effective visibility.

If the goal is to \emph{violate API~2 sensitivity assumptions}, the attack is more subtle. By persisting signals across tracking contexts ($TC$'s), a feedback loop can induce dependencies that break the
assumption that events within each $TC$ are independent. This can
invalidate sensitivity bounds in the \emph{untrusted mode of operation}
(e.g., Privid-style bounding), which assumes computation over short
aggregation intervals are isolated from each other. In contrast, the 
\emph{trusted mode of operation}, which derives sensitivity from the
trusted Object Detection and Tracking pipeline, is not affected by such
cross-context manipulation.

\heading{Summary.}
API~1 does not eliminate all data exposure, but fundamentally changes its nature: from persistent, centralized collection to short-lived, localized observation. Under our threat model, compromise yields only fragmentary views of activity, and reconstructing large-scale behavior requires sustained physical presence or large-scale deployment. This provides a stronger foundation for building an application ecosystem than the conventional centralize-and-retain approach, enabling on-scene applications while avoiding the accumulation patterns that give rise to tracking-based systems.
\section{API 2: Single-Locality Aggregation - Privacy and Utility Analyses}
\label{appendix:api2}


This section formalizes the API~2 \emph{privacy and utility statements} from \S\ref{sec:api2}. We analyze two modes of operation: \emph{trusted event generation}, in which events are produced by the Object Detection and Tracking service (the default), and \emph{untrusted event generation}, in which applications supply their own events. While \sysname supports both, the trusted mode is the default due to the untrusted mode's inflated sensitivity on count queries, which are expected to be common in API~2 (see \S\ref{sec:api2}).

Both modes share the same data model and Laplace release mechanism, but differ in their neighboring relations, how per-context outputs are modeled, and in the precision of per-user accounting. After formalizing the shared data model and standard DP tools (\S\ref{appendix:api2:data-model}), we present per-AW privacy analyses for the two modes (\S\ref{appendix:api2:trusted-model}, \S\ref{appendix:api2:untrusted-model}). We then formalize the binary tree mechanism and shadow bridging release (\S\ref{sec:api2-binary-tree}), prove the node-side privacy guarantees that compose across the tree (\S\ref{appendix:api2:privacy}), and analyze the utility of the DP continual-release mechanism under state reset (\S\ref{appendix:api2:utility}). The utility analysis applies regardless of mode of operation (trusted vs. untrusted), although in practice the bounds differ because sensitivity differs between the two modes.

\subsection{Main Results}
\label{appendix:api1:main-results}

We state the main privacy and utility guarantees upfront. We scope the results to a single locality (node), $\ell$. Let $\{Q_1,\ldots,Q_K\}$ be the registered queries at $\ell$. Each query $Q_q$ has aggregation window $\mathsf{AW}_{Q_q}$, tree size $N_q = \mathsf{maxEC}_{\mathsf{sys}}/\mathsf{AW}_{Q_q}$, per-query sensitivity $\Delta_q$, noise parameter $\sigma_{Q_q}$, and per-tree-node privacy parameter $\varepsilon_q = \Delta_q\sqrt{2}/\sigma_{Q_q}$. The per-query notation ($N$, $\varepsilon$, $\Delta$, $\sigma^2$, the binary tree $\mathcal{B}_N$, the shadow bridging release, $V(N)$, and the neighboring relation $\sim^{\mathsf{user}}_u$) is formalized in \S\ref{sec:api2-binary-tree}. Proofs appear in \S\ref{appendix:api2:privacy} and \S\ref{appendix:api2:utility}.

\begin{theorem}[(User, TC)-level DP on the node]
\label{thm:api2-btm-user-tc}
Consider a locality $\ell$ with registered queries $\{Q_1,\ldots,Q_K\}$. Let $\mathcal{M}$ denote the joint release mechanism across all $K$ queries. Then $\mathcal{M}$ satisfies $(\rho_{\mathsf{node}},0)$-DP at the (user, TC) level, where
\begin{align}
\label{eq:node-filter}
\rho_{\mathsf{node}} := \sum_{q=1}^{K} \Bigl(\log_2 N_q + 2\Bigr)\varepsilon_q.
\end{align}
That is, $\forall u, \forall$ TC $e$ during which $u$ is present, $\forall D \sim^{\mathsf{user}} D'$, $\forall$ measurable $S \subseteq \mathcal{O}$:
\begin{align}
\Pr[\mathcal{M}(D) \in S] \leq e^{\rho_{\mathsf{node}}} \Pr[\mathcal{M}(D') \in S].
\end{align}
\end{theorem}

\begin{theorem}[(User, $\mathsf{maxEC}_{\mathsf{sys}}$)-level composition on the device]
\label{thm:api2-btm-user-maxec}
Fix a user $u$ and one $\mathsf{maxEC}_{\mathsf{sys}}$ window at locality $\ell$. Let $e_1, \ldots, e_G$ ($G \leq \mathsf{maxEC}_{\mathsf{sys}}/\tau_{\mathsf{TC}}$) be the TCs during which $u$ is present, and let $\rho_{\mathsf{node}}^{(g)}$ denote the per-TC filter capacity during $e_g$ (\Thm\ref{thm:api2-btm-user-tc}). The total privacy loss incurred by $u$ is at most $\sum_{g=1}^{G} \rho_{\mathsf{node}}^{(g)}$.
\end{theorem}

Note that to get releases from different $\mathsf{maxEC}_{\mathsf{sys}}$, states will be completely reset independently and the composition happens post-facto, meaning we can just compose them through basic composition here.

\begin{theorem}[Utility of interval-sum queries]
\label{thm:api2-utility}
Let $\sigma^2 = 2\Delta^2/\varepsilon^2$ be the per-node noise variance and $V(N) = 2(\log_2 N - 1)$ for $N\geq 4$.
\begin{enumerate}
    \item For short intervals where $|I|\leq N/2$, the interval sum can be recovered from released tree-node values with noise variance at most $V(N)\cdot\sigma^2$.
    \item For long intervals where $|I| > N/2$, writing $|I| = kN/2 + R$ with $k\geq 1$ and $0\leq R < N/2$, the interval sum can be recovered from released values across containers with variance at most $(k + V(N))\cdot\sigma^2$.
\end{enumerate}
\end{theorem}

\subsection{Data Model}
\label{appendix:api2:data-model}

API 2 organizes time at each locality into three granularities: \emph{tracking contexts} (TCs), \emph{aggregation windows} (AWs), and $\mathsf{maxEC}_{\mathsf{sys}}$ containers. TCs are the smallest unit of time, within which the tracker assigns stable IDs. AWs group one or more TCs and are the unit at which the system produces aggregate outputs. The DP Continual Release service operates over a binary tree whose leaves are AWs, inside an overlapping $\mathsf{maxEC}_{\mathsf{sys}}$ container.

\begin{definition}[Localities and time]
\label{def:api2-localities-time}
Let $\cL$ denote a set of localities. Let $\cT$ be a totally ordered time domain.
\end{definition}

\begin{definition}[Tracking contexts (TC)]
\label{def:api2-tracking-context}
Fix a constant $\tau_{\mathsf{TC}}>0$ (the tracking-context duration).
For each locality $\ell\in \cL$, let $\cE_\ell$ denote a set of half-open time intervals $e=[t_{\mathsf{start}}(e),t_{\mathsf{end}}(e))\subseteq \cT$, each of duration $\tau_{\mathsf{TC}}$, such that:
\begin{enumerate}
    \item the intervals in $\cE_\ell$ are pairwise disjoint, and
    \item $\cE_\ell$ is totally ordered by start time.
\end{enumerate}
Let $\cE := \bigcup_{\ell\in \cL}\cE_\ell$. For $e\in \cE$, write $\mathsf{loc}(e)\in \cL$ for its associated locality.
\end{definition}

The "Object Detection and Tracking" service at locality $\ell$ runs inside a tumbling ephemeral container ($\mathsf{minEC}=0$, $\mathsf{maxEC}=\mathsf{TC}$) that partitions time into consecutive, non-overlapping tracking contexts $\mathsf{tc}_1, \mathsf{tc}_2, \ldots$. Within each TC, the tracker assigns stable pseudonymous IDs to detected individuals. At TC boundaries, all IDs are reset and no tracking state persists.

\begin{definition}[Aggregation windows (AW)]
\label{def:api2-aggregation-windows}
An aggregation window $\mathsf{AW}_Q$ is a contiguous block of one or more consecutive TCs ($\mathsf{AW}_Q$ is a multiple of $\mathsf{TC}$).

The query API specifies $\mathsf{AW}_Q$ as a per-query parameter, subject to the constraint that $N := \mathsf{maxEC}_{\mathsf{sys}}/\mathsf{AW}_Q$ is a power of $2$ (determining the binary tree height $h = \log_2 N$), and that the finest aggregation window equals a single TC.
\end{definition}

Each AW produces one aggregate output that is fed into the continual-release mechanism. The database records are mutually independent across AWs, since each AW's output depends only on the sensing data within that AW.

The privacy analysis relies on the following assumptions.
\begin{itemize}
    \item \textbf{Tracking fidelity.} Within each TC, the tracker assigns each individual a single consistent ID, possibly with a bounded error factor $\rho_{\mathsf{track}}\geq 1$ on the number of IDs per user (which scales the sensitivity). In practice, $\rho_{\mathsf{track}}$ is derived from empirical error rates of the tracker (e.g., BoT-SoRT via YOLO).
    \item \textbf{Broadcast reception.} If a user is within the sensing range of a node (i.e., the user's data may contribute to releases), the user's device receives the node's privacy-loss broadcast. This ensures the device can account for all costs incurred on its behalf.
    \item \textbf{Container ephemerality.} When a $\mathsf{maxEC}_{\mathsf{sys}}$ container is destroyed, all non-DP state (raw data, partial sums, shadow running sums) is irrecoverably deleted. This ensures that releases from distinct $\mathsf{maxEC}_{\mathsf{sys}}$ windows are computed on independent state, which is required for composition across windows (\Thm\ref{thm:api2-btm-user-maxec}).
    \item \textbf{Fixed query workload within each TC.} The set of registered queries, their sensitivities, and their aggregation window parameters are fixed within each TC. This ensures the broadcast cost $\rho_{\mathsf{node}}$ (Eq.~\ref{eq:node-filter}) correctly accounts for all current and future releases involving data from that TC.
\end{itemize}

\begin{definition}[$\mathsf{maxEC}_{\mathsf{sys}}$ container and shadow]
\label{def:api2-maxec-container}
Fix a positive integer $N = 2^h$ ($h\geq 1$). The DP Continual Release service at locality $\ell$ runs inside an overlapping ephemeral container with $\mathsf{maxEC}=\mathsf{maxEC}_{\mathsf{sys}} := N\cdot\mathsf{AW}_Q$ and $\mathsf{minEC}=\mathsf{maxEC}_{\mathsf{sys}}/2$. Primary containers $a^{(0)}_1, a^{(0)}_2, \ldots$ tile time into consecutive blocks of $N$ AWs. A shadow container is initialized $\mathsf{maxEC}_{\mathsf{sys}}/2$ before each primary-container reset. Each primary container runs a binary tree mechanism with $N$ leaves (\S\ref{sec:api2-binary-tree}). The shadow container adds one bridging release at each primary-container boundary.
\end{definition}

The overlapping structure ensures smooth container rotation with $\mathsf{minEC}=\mathsf{maxEC}_{\mathsf{sys}}/2$. The shadow bridging release provides the noisy sum of the second half of each primary container ($N/2$ AWs), connecting adjacent containers for interval-sum recovery. Each AW participates in at most $\log_2 N + 1$ tree-node releases (one per level of the primary tree) plus at most $1$ shadow bridging release, for at most $\log_2 N + 2$ releases total.

The following parameters are treated as public: the locality set $\cL$, the TC schedules $\{\cE_\ell\}$, $\mathsf{maxEC}_{\mathsf{sys}}$, and the per-report sensitivity bounds $\{S_r\}$ (or query templates $\{Q_r\}$ in the trusted class). All of these are deployment-time commitments made before any data is collected. Treating mechanism structure as public is standard in DP \cite{DMNS06}.

\begin{definition}[Per-context scalar output]
\label{def:api2-scalar-output}
Let $\cY := \R_{\geq 0}$ denote the set of per-context outputs. Each aggregation window $e$ produces a single non-negative scalar $y_e \in \cY$. Let $\mathbf{0} := 0\in\cY$.
\end{definition}

For untrusted applications, $y_e$ is an opaque scalar emitted by arbitrary application code, capped at its declared sensitivity. For trusted applications, $y_e$ is the output of a fixed bounded query over trusted structured detections (see \S\ref{appendix:api2:trusted-model}).

\begin{definition}[Database]
\label{def:api2-db}
A database $D$ is a set of locality-context records $x=(\ell,e,y)\in \cX := \cL\times \cE \times \cY$ such that $\ell=\mathsf{loc}(e)$, and if $(\ell,e,y)\in D$ and $(\ell,e,y')\in D$ then $y=y'$. Let $\mathcal{D}$ denote the set of all such databases. For $\ell\in \cL$ and $e\in \cE_\ell$, define $D_\ell^e \in \cY\cup\{\bot\}$ as the unique $y$ with $(\ell,e,y)\in D$ if it exists, and $\bot$ otherwise.
\end{definition}

\subsubsection{Standard DP tools}
\label{appendix:api1:standard-dp-tools}

The following definitions and lemmas are standard and are used by both the trusted and untrusted privacy analyses below.

\begin{definition}[($\epsilon,\delta$)-DP]
\label{def:api2-dp}
A randomized mechanism $\mathcal{M}:\mathcal{D}\to \mathcal{O}$ is $(\epsilon,\delta)$-differentially private with respect to $\sim$ if for all $D\sim D'$ and measurable $S\subseteq\mathcal{O}$:
\begin{align}
\Pr[\mathcal{M}(D)\in S] \leq e^{\epsilon}\Pr[\mathcal{M}(D')\in S] + \delta.
\end{align}
\end{definition}

\begin{lemma}[Basic composition for long-term releases]
\label{lem:api2-basic-composition}
If $\mathcal{M}_t$ is $(\epsilon_t,\delta_t)$-DP for each $t\in[T]$, then the joint mechanism $(\mathcal{M}_1(D),\ldots,\mathcal{M}_T(D))$ is $(\sum_t\epsilon_t,\sum_t\delta_t)$-DP.
\end{lemma}

\begin{lemma}[Group privacy]
\label{lem:api2-group-privacy}
If $\mathcal{M}$ is $(\varepsilon,\delta)$-DP with respect to a neighboring relation $\sim$, and $D,D'$ can be connected by a chain of at most $g$ steps under $\sim$ (i.e., there exist $D = D_0 \sim D_1 \sim \cdots \sim D_g = D'$), then $\Pr[\mathcal{M}(D)\in S] \leq e^{g\varepsilon}\Pr[\mathcal{M}(D')\in S] + \delta_g$, where $\delta_g=\sum_{j=0}^{g-1}e^{j\varepsilon}\delta$. In particular, if $\delta=0$ then $\mathcal{M}$ satisfies $(g\varepsilon,0)$-DP between $D$ and $D'$.
\end{lemma}

\subsection{Privacy Analysis for Trusted Mode}
\label{appendix:api2:trusted-model}

In the trusted class, a pre-installed trusted object detector processes raw sensor data and emits a structured detection record for each context before any application code runs. Applications are restricted to fixed, bounded query templates (e.g. counts and sums) over those detections. In the future, there may be opportunities to enrich the queries that can be supported via prior works on Privacy-Preserving SQL queries such as those studied in \cite{cangialosi2022privid, grislain2024qrlew, zhang2023differentially}. Sensitivity is derived from the query type, so noise is calibrated to the actual per-user contribution rather than a worst-case declared bound.

Each trusted release corresponds to exactly one AW and there is no multi-context accumulation. The trusted class can release at this fine granularity because the per-user sensitivity $\mu_Q$ is typically small (e.g., $\mu_Q=1$ for count queries). The detector assigns a pseudonymous ID $\mathsf{id}_u$ to each detected user within each TC (\Def\ref{def:api2-tracking-context}), so each user's contribution within a TC is bounded. Each user in the detection record $\cD_e$ (defined formally in \Def\ref{def:api2-trusted-query} below) carries an attribute value $v\in[0,v_{\max}]$. Removing or adding one user changes a $\operatorname{COUNT}$ query by $1$ and a $\operatorname{SUM}$ query by at most $v_{\max}$, giving per-user sensitivity $\mu_Q$.

\subsubsection{Trusted data model and sensitivity}
\label{sec:api2-trusted-privacy}

Note that for both trusted and untrusted classes, context isolation holds based on API~1. That is, application code has no persistent state across TCs. This system property was justified separately by the design of API 1, so we only discuss the aggregation granularity here.

\begin{definition}[User]
\label{def:api2-user}
A \emph{user} is a detected presence identified by a pseudonymous ID $\mathsf{id}_u$ assigned by the trusted object detector within a single TC (\Def\ref{def:api2-tracking-context}). IDs are reset at TC boundaries. $\cD_e = (F_1, F_2, \ldots, F_m)$ is the ordered sequence of frame detection sets produced during AW $e = [t_{\mathsf{start}}(e), t_{\mathsf{end}}(e))$, where each $F_i$ contains the entities detected at that instant.
\end{definition}

A user who spans two TCs receives distinct IDs in each and is treated as two separate presences for privacy purposes. Under the "tracking fidelity" assumption, a user has exactly one ID per TC, so within an AW spanning $k_e$ TCs, a user has at most $k_e$ IDs.

\begin{definition}[Trusted detection record and query template]
\label{def:api2-trusted-query}
A detection record for AW $e$ is the frame detection sequence $\cD_e = (F_1, F_2, \ldots, F_m)$ from \Def\ref{def:api2-user}, where each $F_i = \{(\mathsf{id},\tau,v)\}$ records the pseudonymous IDs, object types $\tau$, and scalar attributes $v\in[0,v_{\max}]$ ($v_{\max}>0$ declared per attribute) of entities detected at sensor frame $i$. A trusted query template $Q_r$ is an aggregation function over $\cD_e$. The choice of frame subset and aggregation method determines the per-ID sensitivity $\mu_Q$ (\Def\ref{def:api2-trusted-sensitivity}).
\end{definition}

Representative query templates include $\operatorname{COUNT}(\tau)$ (number of type-$\tau$ entities in a designated frame subset) and $\operatorname{SUM}(\tau,v)$ (sum of attribute $v$ over type-$\tau$ entities). Averages are represented as $\operatorname{SUM}/\operatorname{COUNT}$ and treated as a scaled sum for DP purposes.

\begin{definition}[Trusted query sensitivity]
\label{def:api2-trusted-sensitivity}
The per-ID sensitivity $\mu_Q$ of a trusted query template $Q_r$ is the maximum change in the query output from adding or removing all occurrences of a single tracked ID across the frame sequence:
\begin{align}
\mu_Q := \max_{\cD_e,\mathsf{id}} \Bigl|Q_r(\cD_e) - Q_r(\cD_e^{\mathsf{id}\to\emptyset})\Bigr|,
\end{align}
where $\cD_e^{\mathsf{id}\to\emptyset}$ is $\cD_e$ with all entries for $\mathsf{id}$ removed from every frame.
\end{definition}
For example, $\mu_Q = 1$ for last-frame $\operatorname{COUNT}$ queries, $\mu_Q = v_{\max}$ for last-frame $\operatorname{SUM}$ queries, and $\mu_Q$ scales with frame count for cumulative aggregations (e.g., $\mu_Q = k$ for a $k$-frame cumulative count). The value of $\mu_Q$ is declared at query registration.

\paragraph{Neighboring relations.} The trusted class uses two neighboring relations that serve different roles. $D\sim^{\mathsf{id}} D'$ holds iff $\cD_e$ and $\cD'_e$ differ by exactly one tracked ID for some AW $e$, with all other per-context records unchanged. Noise is calibrated to this ID-level relation, giving per-ID sensitivity $\mu_{Q_r}$ under $\sim^{\mathsf{id}}$ (formalized in \Def\ref{def:api2-trusted-report}).

$D\sim^{\mathsf{user}} D'$ holds iff $D'$ is obtained from $D$ by removing all IDs attributed to one user $u$ within a single AW. This is the semantically meaningful guarantee on the node. Under the "tracking fidelity" assumption with $\rho_{\mathsf{track}}=1$, a user receives exactly one ID per TC. If AW $e$ spans $k_e$ TCs, user $u$ has at most $k_e$ distinct IDs within $e$ (one per TC). Removing $u$ from $e$ under $\sim^{\mathsf{user}}$ amounts to removing $k_e$ IDs under $\sim^{\mathsf{id}}$. When $\mathsf{AW}_Q = \mathsf{TC}$ (the finest granularity), $k_e = 1$ and the two neighboring relations coincide.

This section presents an idealized version of the trusted-class analysis. The evaluation in the main text examines the effect of tracking errors empirically (e.g., when multiple IDs are assigned to a single continuous presence of a user within an AW).

\begin{definition}[Trusted API 2 report ($k_r=1$)]
\label{def:api2-trusted-report}
A trusted report $r$ is associated with a locality $\ell_r$, a single AW $e_r$ (so $k_r=1$), and a trusted query template $Q_r$. The trusted report function is $\rho_r^{\mathsf{tr}}(D) := Q_r(\cD_{e_r})$. Under $\sim^{\mathsf{id}}$, $\Delta(\rho_r^{\mathsf{tr}})=\mu_{Q_r}$ by \Def\ref{def:api2-trusted-sensitivity}.
\end{definition}

\begin{theorem}[DP via Laplace noise (trusted)]
\label{thm:api2-dp-laplace-trusted}
The mechanism
\begin{align}
M^{\mathsf{tr}}(D) := \rho_r^{\mathsf{tr}}(D) + \eta, \quad \eta\sim\mathrm{Lap}(\mu_{Q_r}/\varepsilon),
\end{align}
is $(\varepsilon,0)$-DP under $\sim^{\mathsf{id}}$. Under $\sim^{\mathsf{user}}$, if AW $e$ spans $k_e$ TCs, a user $u$ present in all $k_e$ TCs incurs privacy cost at most $k_e\varepsilon$ for the single-AW release.
\end{theorem}

\begin{proof}
$\Delta(\rho_r^{\mathsf{tr}})=\mu_{Q_r}$ under $\sim^{\mathsf{id}}$ by \Def\ref{def:api2-trusted-sensitivity}. The Laplace mechanism with scale $\mu_{Q_r}/\varepsilon$ gives $(\varepsilon,0)$-DP under $\sim^{\mathsf{id}}$.

For $\sim^{\mathsf{user}}$: let AW $e$ span $k_e$ TCs. By tracking fidelity, $u$ has at most $k_e$ IDs in $e$ (one per TC). Removing $u$ from $e$ under $\sim^{\mathsf{user}}$ is a $k_e$-step chain under $\sim^{\mathsf{id}}$:
\begin{align}
D = D_0 \sim^{\mathsf{id}} D_1 \sim^{\mathsf{id}} \cdots \sim^{\mathsf{id}} D_{k_e} = D'.
\end{align}
By \Lem\ref{lem:api2-group-privacy} with $g = k_e$ and $\delta = 0$, the cost is $k_e\varepsilon$. When $\mathsf{AW}_Q = \mathsf{TC}$, $k_e = 1$.
\end{proof}

In the trusted class, the detector identifies which users contributed to each AW's output, so privacy loss can be attributed per user. In the untrusted class, all users present in the AW must be charged pessimistically since the application output may depend on any of them. The total per-user privacy cost across the binary tree is analyzed in \Thm\ref{thm:api2-btm-user-tc} and \Thm\ref{thm:api2-btm-user-maxec}.

\subsubsection{Accounting protocol and algorithms}
\label{appendix:api2:trusted-accounting}

The node continuously broadcasts the current TC identifier and the per-TC privacy cost $\rho_{\mathsf{node}}$ (Eq.~\ref{eq:node-filter}), which accounts for all registered queries. Note that the broadcast is not per-ID, since every device in range receives the same $(\ell, e, \rho_{\mathsf{node}})$ tuple and accounts for the same loss (again, the first two elements merely identify the TC, and the third one is for privacy accounting). Per this design, the node doesn't need to know which device corresponds to which tracked ID.

\begin{definition}[Privacy broadcast]
\label{def:api2-privacy-broadcast}
Throughout each TC $e$, the \sysname node at locality $\ell$ continuously broadcasts $(\ell, e, \rho_{\mathsf{node}})$. Every device in range receives the same broadcast. A device that is present at any point during $e$ records $\rho_{\mathsf{node}}$ once (deduplicating by the TC identifier $e$). The cost $\rho_{\mathsf{node}}$ pre-charges for all tree-node releases and shadow bridging releases across all queries that will ever involve TC $e$'s data.
\end{definition}

Alg.~\ref{alg:api2-accounting} shows the release and accounting procedure for the trusted class. Alg.~\ref{alg:api2-trusted-detection} shows how the trusted detector builds $\cD_e$ (\Def\ref{def:api2-trusted-query}) and produces DP releases. \textsc{ProcessFrame} runs at sensor rate, maintaining the set of active IDs for query evaluation. \textsc{AWRelease} is called at $t_{\mathsf{end}}(e)$, evaluating $Q_r$ on the full frame sequence $\cD_e$.

\begin{algorithm}[t]
\caption{Release and accounting at locality $\ell$ (trusted, binary tree)}
\label{alg:api2-accounting}
\algrenewcommand\algorithmicfunction{\textbf{def}}
\begin{algorithmic}[1]
\State \graycomment{\sysname{} node at locality $\ell$}
\State $\varepsilon > 0$ \graycomment{per-tree-node privacy cost}
\State $N > 0$ \graycomment{AWs per container (\Def\ref{def:api2-maxec-container})}
\State $\mathsf{Tree} \gets \textsc{NewBinaryTree}(N)$ \graycomment{$\cB_N$ as in \S\ref{sec:api2-binary-tree}}
\State $\mathsf{ShadowSum} \gets 0$ \graycomment{bridging sum for shadow container}
\State $\mathsf{leafIdx} \gets 0$ \graycomment{leaf counter within current container}
\Statex
\Function{OnAW}{$e, Q_r, \mu_{Q_r}$}
  \State $y_e \gets Q_r(\cD_e)$ \graycomment{query over detection record (\Def\ref{def:api2-trusted-query})}
  \State $\mathsf{Tree}.\textsc{AddLeaf}(y_e)$ \graycomment{insert $y_e$ as next leaf of $\cB_N$}
  \State $\mathsf{leafIdx} \gets \mathsf{leafIdx} + 1$
  \If{$\mathsf{leafIdx} > N/2$} \graycomment{second half of container}
    \State $\mathsf{ShadowSum} \gets \mathsf{ShadowSum} + y_e$
  \EndIf
  \ForAll{$v \in \mathsf{Tree}$ whose subtree rooted at $v$ just became complete} \graycomment{release per Eq.~\ref{eq:tree-node-release}}
    \State $\eta_v \sim \mathrm{Lap}(\mu_{Q_r}/\varepsilon)$
    \State \textsc{CloudRelease}($\ell, v, \mathsf{sum}(v) + \eta_v$)
  \EndFor
  \State $\mathsf{Broadcast}(\ell, e, (\log_2 N + 2)\varepsilon)$ \graycomment{\Def\ref{def:api2-privacy-broadcast}}
\EndFunction
\Statex
\Function{OnContainerEnd}{} \graycomment{shadow bridging release}
  \State $\eta \sim \mathrm{Lap}(\mu_{Q_r}/\varepsilon)$
  \State \textsc{CloudRelease}($\ell, \mathsf{shadow}, \mathsf{ShadowSum} + \eta$)
  \State \textsc{Teardown}($\mathsf{Tree}$) \graycomment{destroy all raw state ("container ephemerality")}
  \State $\mathsf{Tree} \gets \textsc{NewBinaryTree}(N)$
  \State $\mathsf{ShadowSum} \gets 0$
  \State $\mathsf{leafIdx} \gets 0$
\EndFunction
\Statex
\State \graycomment{User's device}
\State $\varepsilon_{\mathsf{acc}} \gets 0$ \graycomment{accumulated privacy loss}
\Function{Receive}{$\ell, e, \rho$}
  \State $\varepsilon_{\mathsf{acc}} \gets \varepsilon_{\mathsf{acc}} + \rho$
\EndFunction
\end{algorithmic}
\end{algorithm}

\begin{algorithm}[t]
\caption{Detection and release by \sysname{} locality $\ell$ (trusted)}
\label{alg:api2-trusted-detection}
\algrenewcommand\algorithmicfunction{\textbf{def}}
\begin{algorithmic}[1]
\State \graycomment{\sysname{} node at locality $\ell$}
\State $\cD_e \gets [\ ]$ \graycomment{frame sequence for current AW $e$}
\State $\mathsf{ActiveIDs} \subseteq \mathsf{Id}$, init $\emptyset$
\State $\mathsf{Tracker}$ \graycomment{multi-object tracker}
\Statex
\Function{ProcessFrame}{$\mathsf{frame}_i$}
  \State $\mathsf{rawDets} \gets \textsc{Detector.Run}(\mathsf{frame}_i)$
  \State $F_i \gets \textsc{Tracker.Assign}(\mathsf{rawDets})$
  \State $\mathsf{ActiveIDs} \gets \mathsf{ActiveIDs} \cup F_i$
  \State $\mathsf{ActiveIDs} \gets \mathsf{ActiveIDs} \setminus \{\mathsf{id} : \text{absent} \geq \tau_{\mathsf{exit}}\}$
  \State $\cD_e.\textsc{Append}(F_i)$
\EndFunction
\Statex
\Function{AWRelease}{$e, Q_r, \varepsilon$}
  \State $\mathsf{result} \gets Q_r(\cD_e)$
  \State $\eta \sim \mathrm{Lap}(\mu_{Q_r} / \varepsilon)$
  \State \textsc{CloudRelease}($\ell, e, \mathsf{result} + \eta$)
  \State $\mathsf{Broadcast}(\ell, e, (\log_2 N + 2)\varepsilon)$
  \State $\cD_e \gets [\ ]$
\EndFunction
\end{algorithmic}
\end{algorithm}

A user's device accumulates $\varepsilon_{\mathsf{acc}} = G \cdot \rho$ where $G$ is the number of distinct AWs in which it received the broadcast (the "broadcast reception" assumption ensures reception whenever the user is in sensing range). Every user present during an AW accounts for the same $\rho$, regardless of whether their data actually affected the release. The node-side mechanism does not adapt to individual budgets.

\subsection{Privacy Analysis for Untrusted Mode}
\label{appendix:api2:untrusted-model}

\subsubsection{Query model}
\label{sec:api2-untrusted-model}
In the untrusted mode, the application code that produces events is treated as arbitrary, possibly adversarial. The application receives raw sensing data for each context and emits a single scalar $y_e\in[0,S]$ at the end of each AW. The system does not inspect intermediate computations or per-user outputs. The application declares a sensitivity bound $S > 0$, which the runtime enforces as a hard cap on each per-context output and uses as the noise scale $S/\varepsilon$ (see \Thm\ref{thm:api2-dp-laplace}).

\begin{definition}[Declared sensitivity bound and output cap]
\label{def:api2-declared-sensitivity}
An untrusted application declares a sensitivity bound $S>0$. The per-context output is
\begin{align}
y_e := \min\left\{\mathsf{Ext}_e(\mathsf{Raw}_e), S\right\},
\end{align}
where $\mathsf{Ext}_e$ is the application's extraction function over the raw sensor content $\mathsf{Raw}_e$ of context $e$. Throughout, $y_e$ denotes the capped value, so $y_e\in[0,S]$.
\end{definition}

Each AW produces one capped scalar $y_e\in[0,S]$ (\Def\ref{def:api2-declared-sensitivity}). The per-AW output is fed directly into the binary tree mechanism (\S\ref{sec:api2-binary-tree}), which accumulates and releases noisy partial sums over contiguous blocks of AWs. The global sensitivity of any single AW's contribution is $S$.

\subsubsection{Neighboring relation and sensitivity}

\begin{definition}[Neighboring databases]
\label{def:api2-neighboring}
We write $D\sim_x D'$, for $x=(\ell,e,y)\in \cX$, iff $D'$ is obtained from $D$ by zeroing out the record $(\ell,e)$:
\begin{align}
D' = \left(D\setminus\{(\ell,e,y)\}\right)\cup \{(\ell,e,\mathbf{0})\}.
\end{align}
Write $D\sim D'$ if $D\sim_x D'$ for some $x\in\cX$.
\end{definition}

\begin{definition}[Sensitivity]
\label{def:api2-sensitivity}
For a report $\rho:\mathcal{D}\to \R$ and $x\in\cX$, define the individual sensitivity $\Delta_x(\rho) := \max_{D\sim_x D'} |\rho(D)-\rho(D')|$ and the global sensitivity $\Delta(\rho) := \max_{x\in\cX}\Delta_x(\rho)$.
\end{definition}

\begin{lemma}[Global sensitivity of untrusted API 2 reports]
\label{lem:api2-sensitivity-exact-formula}
Fix an untrusted API 2 report $\rho_r$ with declared sensitivity $S_r$. Then
\begin{align}
\Delta(\rho_r) = S_r.
\end{align}
\end{lemma}

\begin{proof}
Neighboring databases differ by zeroing out one context record, replacing one summand $y_j\in[0,S_r]$ by $0$. Hence $|\rho_r(D)-\rho_r(D')|=y_j\leq S_r$. The bound is achieved by $y_j=S_r$.
\end{proof}

The individual sensitivity $\Delta_x(\rho_r)$ equals $y$ if $x=(\ell,e,y)$ corresponds to the report's AW, and $0$ otherwise.

\subsubsection{DP of untrusted API 2 releases}

\begin{theorem}[DP via Laplace noise (untrusted)]
\label{thm:api2-dp-laplace}
Fix an untrusted API 2 report $\rho_r$ with declared sensitivity $S_r$. The mechanism
\begin{align}
M(D) := \rho_r(D) + \eta, \quad \eta\sim\mathrm{Lap}(S_r/\varepsilon),
\end{align}
is $(\varepsilon,0)$-DP.
\end{theorem}

The noise scale depends only on the declared sensitivity $S_r$, so the system does not need to inspect the application's internal computation.

Alg.~\ref{alg:api2-untrusted-runtime} shows the runtime procedure. Each AW's capped output is fed into the primary container's binary tree. AWs in the second half of the container additionally accumulate into the shadow's running sum, which produces one bridging release at the container boundary. The node continuously broadcasts $(\ell, e, \rho)$ (\Def\ref{def:api2-privacy-broadcast}) so that every device in range accounts for the same loss $\rho = (\log_2 N + 2)\varepsilon$. All users are charged pessimistically since the opaque application output may depend on any of them.

\begin{algorithm}[t]
\caption{Release at locality $\ell$ (untrusted, binary tree)}
\label{alg:api2-untrusted-runtime}
\algrenewcommand\algorithmicfunction{\textbf{def}}
\begin{algorithmic}[1]
\State \graycomment{\sysname{} node at locality $\ell$}
\State $S > 0$ \graycomment{declared sensitivity (\Def\ref{def:api2-declared-sensitivity})}
\State $\varepsilon > 0$ \graycomment{per-tree-node privacy cost}
\State $N > 0$ \graycomment{AWs per container (\Def\ref{def:api2-maxec-container})}
\State $\mathsf{Tree} \gets \textsc{NewBinaryTree}(N)$ \graycomment{$\cB_N$ as in \S\ref{sec:api2-binary-tree}}
\State $\mathsf{ShadowSum} \gets 0$ \graycomment{bridging sum for shadow container}
\State $\mathsf{leafIdx} \gets 0$ \graycomment{leaf counter within current container}
\Statex
\Function{OnAW}{$e$}
  \State $y_e \gets \min(\textsc{RunApp}(\mathsf{Raw}_e), S)$
  \State $\mathsf{Tree}.\textsc{AddLeaf}(y_e)$ \graycomment{insert $y_e$ as next leaf of $\cB_N$}
  \State $\mathsf{leafIdx} \gets \mathsf{leafIdx} + 1$
  \If{$\mathsf{leafIdx} > N/2$} \graycomment{second half of container}
    \State $\mathsf{ShadowSum} \gets \mathsf{ShadowSum} + y_e$
  \EndIf
  \ForAll{$v \in \mathsf{Tree}$ whose subtree rooted at $v$ just became complete} \graycomment{release per Eq.~\ref{eq:tree-node-release}}
    \State $\eta_v \sim \mathrm{Lap}(S/\varepsilon)$
    \State \textsc{CloudRelease}($\ell, v, \mathsf{sum}(v) + \eta_v$)
  \EndFor
  \State $\mathsf{Broadcast}(\ell, e, (\log_2 N + 2)\varepsilon)$ \graycomment{\Def\ref{def:api2-privacy-broadcast}}
\EndFunction
\Statex
\Function{OnContainerEnd}{} \graycomment{shadow bridging release}
  \State $\eta \sim \mathrm{Lap}(S/\varepsilon)$
  \State \textsc{CloudRelease}($\ell, \mathsf{shadow}, \mathsf{ShadowSum} + \eta$)
  \State \textsc{Teardown}($\mathsf{Tree}$) \graycomment{destroy all raw state ("container ephemerality")}
  \State $\mathsf{Tree} \gets \textsc{NewBinaryTree}(N)$
  \State $\mathsf{ShadowSum} \gets 0$
  \State $\mathsf{leafIdx} \gets 0$
\EndFunction
\end{algorithmic}
\end{algorithm}

The per-AW DP guarantees for both modes are now established: $(\varepsilon,0)$-DP per release under $\sim^{\mathsf{id}}$ (trusted) or $\sim$ (untrusted), with sensitivity $\Delta = \mu_Q$ or $\Delta = S_r$ respectively. The next two sections compose these per-AW guarantees across the binary tree and shadow to obtain the node-level privacy and utility bounds stated in Theorems~\ref{thm:api2-btm-user-tc}-\ref{thm:api2-utility}.

\subsection{Binary Tree Mechanism}
\label{sec:api2-binary-tree}

This section formalizes the binary tree mechanism within the $\mathsf{maxEC}_{\mathsf{sys}}$ container (\Def\ref{def:api2-maxec-container}). Write $N := \mathsf{maxEC}_{\mathsf{sys}}/\mathsf{AW}_Q$ for the number of AW leaves per primary container, $h := \log_2 N$, where $N = 2^h$, $h\geq 1$, $\Delta\in\{S_r,\mu_Q\}$ for the per-AW sensitivity (depending on it being trusted or untrusted), and $\sigma^2 := \mathrm{Var}(\mathrm{Lap}(\Delta/\varepsilon)) = 2\Delta^2/\varepsilon^2$.

\paragraph{The base case.}
We first deal with the simpler case where the queried interval lies entirely within one primary container.
Fix a locality $\ell$ and one primary container $a^{(0)}_k$ spanning $N$ consecutive AWs $e_1,\ldots,e_N$. Build a complete binary tree $\cB_N$ whose $N$ leaves are $e_1,\ldots,e_N$. Each node $v$ at depth $d$ (root at depth $0$, leaves at depth $h$) represents a contiguous block of $2^{h-d}$ consecutive leaves. Define $\mathsf{sum}(v) := \sum_{t\in\mathsf{leaves}(v)} y_t$. At the end of each AW $e_t$, the mechanism releases noisy partial sums for every node $v$ such that $e_t$ is the last leaf in the subtree rooted at $v$ (i.e., the subtree rooted at $v$ just became complete):
\begin{align}
\label{eq:tree-node-release}
\tilde{s}(v) := \mathsf{sum}(v) + \eta_v, \quad \eta_v \sim \mathrm{Lap}(\Delta/\varepsilon),
\end{align}
with all $2N-1$ noise draws being independent. This is the standard binary mechanism of \cite{chan2011private}.

\paragraph{Shadow container.}
A shadow container is initialized $N/2$ AWs before each primary container reset (\Def\ref{def:api2-maxec-container}). The shadow maintains a running sum over the second half of the primary container (AWs $N/2+1$ through $N$). At the primary-container boundary (position $kN$), the shadow releases one additional noisy partial sum, perturbed by $\mathrm{Lap}(\Delta/\varepsilon)$. This single bridging release connects two adjacent primary containers, enabling interval-sum recovery across primary-container boundaries. Each AW in the second half of a container participates in $\log_2 N + 1$ releases from the primary tree plus $1$ bridging release from the shadow, for a total of $\log_2 N + 2$ releases. AWs in the first half participate in $\log_2 N + 1$ tree releases only, so $\log_2 N + 2$ is a worst-case bound that applies uniformly to all AWs.

\paragraph{Longer queried intervals via decomposition and induction.} Next, we build the tools for handling queries whose intervals (1) span a primary-container boundary, (2) cover more than $N$ AWs, or both.

\begin{definition}[Canonical interval decomposition]
\label{def:api2-canonical}
For a contiguous sub-interval $[i,j]\subseteq [1,N]$, the canonical decomposition $\mathsf{Can}([i,j])$ is the unique minimal set of nodes in $\cB_N$ whose leaf sets partition $\{i,\ldots,j\}$. Define $V(N) := 2(h-1)$ for $N\geq 4$ and $V(2):=1$.
\end{definition}

\begin{lemma}[Prefix and suffix canonical sizes]
\label{lem:btm-prefix-suffix}
Let $\mathsf{popcount}(n)$ denote the number of $1$-bits in the binary representation of a positive integer $n$.
\begin{enumerate}
    \item A prefix $[1,r]$ in $\cB_N$ has $|\mathsf{Can}([1,r])| = \mathsf{popcount}(r)$.
    \item A suffix $[N-\ell+1, N]$ in $\cB_N$ has $|\mathsf{Can}([N-\ell+1,N])| = \mathsf{popcount}(\ell)$.
    \item For any $[i,j]\subseteq[1,N]$, $|\mathsf{Can}([i,j])|\leq V(N)$.
\end{enumerate}
\end{lemma}

\begin{proof}
(1) By induction on $h$. Base case $h=1$: $r\in\{1,2\}$, $\mathsf{Can}([1,r])$ consists of one node, and $\mathsf{popcount}(r)=1$. For the inductive step, if $r\leq N/2$, then $[1,r]$ lies in the left subtree and $|\mathsf{Can}([1,r])| = \mathsf{popcount}(r)$ by induction. If $r > N/2$, write $r = N/2 + r'$. The prefix decomposes as the left child of the root (one node) plus the prefix $[1,r']$ in the right subtree ($\mathsf{popcount}(r')$ nodes by induction). Total: $1 + \mathsf{popcount}(r') = \mathsf{popcount}(N/2 + r') = \mathsf{popcount}(r)$.

(2) By symmetry: reflect $\cB_N$ (swap left/right children at every level). The suffix $[N-\ell+1,N]$ maps to the prefix $[1,\ell]$ in the reflected tree. By part~(1), this has $\mathsf{popcount}(\ell)$ nodes.

(3) If $i\leq N/2 < j$, write $[i,j] = [i,N/2]\cup[N/2+1,j]$. The first piece is a suffix of the left subtree of length $\ell' = N/2-i+1$, the second is a prefix of the right subtree of length $r' = j-N/2$. Both $\ell',r'\leq N/2$, so each popcount is at most $h-1$. Total: at most $2(h-1) = V(N)$. The other cases ($j\leq N/2$ or $i > N/2$) follow by induction.
\end{proof}

\subsection{End-to-End Privacy Guarantee Enforced on Each Node}
\label{appendix:api2:privacy}

We first establish a per-query privacy guarantee at the $\mathsf{AW}_Q$ level, then compose across all queries at TC granularity to prove \Thm\ref{thm:api2-btm-user-tc} and \Thm\ref{thm:api2-btm-user-maxec}.

\begin{lemma}[Per-query, per-AW DP under $\sim^{\mathsf{id}}$]
\label{lem:api2-per-query-dp}
Fix a single query $Q_q$ with aggregation window $\mathsf{AW}_{Q_q}$, tree size $N_q = \mathsf{maxEC}_{\mathsf{sys}}/\mathsf{AW}_{Q_q}$, and per-tree-node privacy parameter $\varepsilon_q$. Fix an AW $e$ of $Q_q$. Then the joint release of all noisy tree-node partial sums and the shadow bridging release \emph{for query $Q_q$} involving $e$'s data is $((\log_2 N_q + 2)\varepsilon_q, 0)$-DP under $\sim^{\mathsf{id}}$ (i.e., with respect to adding or removing a single tracked ID in $e$).
\end{lemma}

\begin{proof}
Let $D \sim^{\mathsf{id}} D'$ differ by one tracked ID in AW $e$. Then $|y_e(D) - y_e(D')| \leq \Delta_q$ and $y_{e'}(D) = y_{e'}(D')$ for all $e' \neq e$. Let $\{M_j\}_{j=1}^{L}$ with $L := \log_2 N_q + 2$ denote the releases involving $e$: $\log_2 N_q + 1$ tree-node sums plus at most $1$ shadow bridging release. Each $M_j$ adds independent $\mathrm{Lap}(\Delta_q/\varepsilon_q)$ noise. By the Laplace mechanism, each $M_j$ satisfies $(\varepsilon_q, 0)$-DP under $\sim^{\mathsf{id}}$. The noise draws are mutually independent, so by \Lem\ref{lem:api2-basic-composition}, we get that $(M_1, \ldots, M_L)$ is
\begin{align}
(L \cdot \varepsilon_q, 0)\text{-DP} = \Bigl((\log_2 N_q + 2)\varepsilon_q, 0\Bigr)\text{-DP.}
\end{align}
\end{proof}

\noindent\textbf{Theorem \ref{thm:api2-btm-user-tc}} (restated)\textbf{.}
\textit{The node enforces a global privacy filter with per-TC capacity $\rho_{\mathsf{node}} := \sum_{q=1}^{K} (\log_2 N_q + 2)\varepsilon_q$. For any user $u$ present during TC $e$, the joint release across all queries involving $e$'s data is $(\rho_{\mathsf{node}}, 0)$-DP with respect to $D \sim^{\mathsf{user}}_u D'$.}

\begin{proof}
Fix TC $e$ and user $u$ present during $e$. By tracking fidelity, $u$ has one tracked ID $\mathsf{id}_u$ within $e$. For each $q \in [K]$, let $e_q$ denote the unique AW of $Q_q$ containing $e$ (since $\mathsf{AW}_{Q_q}/\tau_{\mathsf{TC}} \in \mathbb{N}$). Let $\mathcal{M}_q$ denote the joint release of $Q_q$ for AW $e_q$. Under $D \sim^{\mathsf{user}}_u D'$, removing $u$ from $e$ removes $\mathsf{id}_u$ from $e_q$, i.e., one step under $\sim^{\mathsf{id}}$. By \Lem\ref{lem:api2-per-query-dp}, $\mathcal{M}_q$ is $((\log_2 N_q + 2)\varepsilon_q, 0)$-DP. The $\{\mathcal{M}_q\}_{q=1}^{K}$ use independent noise (separate trees), so by \Lem\ref{lem:api2-basic-composition}, we get that $(\mathcal{M}_1, \ldots, \mathcal{M}_K)$ is
\begin{align}
\Bigl(\textstyle\sum_{q=1}^{K}(\log_2 N_q + 2)\varepsilon_q, 0\Bigr)\text{-DP} = (\rho_{\mathsf{node}}, 0)\text{-DP.}
\end{align}
\end{proof}

\noindent\textbf{Theorem \ref{thm:api2-btm-user-maxec}} (restated)\textbf{.}
\textit{Fix a user $u$ and one $\mathsf{maxEC}_{\mathsf{sys}}$ window at locality $\ell$. Let $e_1, \ldots, e_G$ ($G \leq \mathsf{maxEC}_{\mathsf{sys}}/\tau_{\mathsf{TC}}$) be the TCs during which $u$ is present, with per-TC filter capacity $\rho_{\mathsf{node}}^{(g)}$. The total privacy loss is at most $\sum_{g=1}^{G} \rho_{\mathsf{node}}^{(g)}$. Releases from distinct $\mathsf{maxEC}_{\mathsf{sys}}$ windows are computed on independent state and compose additively.}

\begin{proof}
By \Thm\ref{thm:api2-btm-user-tc}, the joint release for each $e_g$ is $(\rho_{\mathsf{node}}^{(g)}, 0)$-DP under $\sim^{\mathsf{user}}_u$ restricted to $e_g$. Removing $u$ from all $G$ TCs yields $D = D_0 \sim^{\mathsf{user}}_u D_1 \sim^{\mathsf{user}}_u \cdots \sim^{\mathsf{user}}_u D_G = D'$. By \Lem\ref{lem:api2-basic-composition} with $\delta = 0$, we get that the total privacy cost is
\begin{align}
\sum_{g=1}^{G} \rho_{\mathsf{node}}^{(g)}.
\end{align}
\end{proof}

\paragraph{Proposed device-side accounting.} At the beginning of each TC $e_g$, the node computes $\rho_{\mathsf{node}}^{(g)}$ from the currently registered queries and continuously broadcasts $(\ell, e_g, \rho_{\mathsf{node}}^{(g)})$ (\Def\ref{def:api2-privacy-broadcast}). A device in range records $\rho_{\mathsf{node}}^{(g)}$ once per TC and accumulates $\sum_{g} \rho_{\mathsf{node}}^{(g)}$ over the TCs during which it is present.

\subsection{Utility of Binary Tree Mechanism with State Reset}
\label{appendix:api2:utility}

We restate and prove \Thm\ref{thm:api2-utility}. Recall from \Def\ref{def:api2-canonical} that $V(N) := 2(h-1) = 2(\log_2 N - 1)$ for $N\geq 4$ is the worst-case number of tree nodes needed to decompose any sub-interval of $[1,N]$.

\noindent\textbf{Theorem \ref{thm:api2-utility}} (restated)\textbf{.}
\textit{Let $\sigma^2 = 2\Delta^2/\varepsilon^2$ be the per-node noise variance.
(1) For short intervals where $|I|\leq N/2$, the interval sum can be recovered from released tree-node values with noise variance at most $V(N)\cdot\sigma^2$.
(2) For long intervals where $|I| = kN/2 + R$ with $k\geq 1$ and $0\leq R < N/2$, the interval sum can be recovered with variance at most $(k + V(N))\cdot\sigma^2$.}

\begin{proof}
\emph{Case 1: $|I| \leq N/2$.} If $I$ lies within one primary container, by (3) of \Lem\ref{lem:btm-prefix-suffix}, we have $|\mathsf{Can}(I)| \leq V(N)$, so
\begin{align}
\mathrm{Var}[\hat{S}(I)] \leq V(N) \cdot \sigma^2.
\end{align}
If $I$ straddles a boundary, write $I = I_1 \cup I_2$ with $I_1$ a suffix of one container ($|I_1| = s$) and $I_2$ a prefix of the next ($|I_2| = p$), $s + p \leq N/2$. By \Lem\ref{lem:btm-prefix-suffix}(1-2):
\begin{align}
|\mathsf{Can}(I_1)| + |\mathsf{Can}(I_2)|
    &= \mathsf{popcount}(s) + \mathsf{popcount}(p)\\
    &\leq 2(h-1) = V(N).
\end{align}
All noise draws across the two trees are independent, giving $\mathrm{Var}[\hat{S}(I)] \leq V(N) \cdot \sigma^2$.

\emph{Case 2: $|I| > N/2$.} Write $|I| = kN/2 + R$ with $k \geq 1$, $0 \leq R < N/2$. The depth-1 nodes tile the timeline into aligned $N/2$-blocks, each answered by one tree node with variance $\sigma^2$. Decompose $I$ into at most $k$ complete blocks plus two residual pieces (a suffix and a prefix, each $< N/2$). By (1-2) of \Lem\ref{lem:btm-prefix-suffix}:
\begin{align}
\text{total nodes}
    &\leq k + \mathsf{popcount}(s) + \mathsf{popcount}(p)\\
    &\leq k + V(N).
\end{align}
All noise draws are independent, so $\mathrm{Var}[\hat{S}(I)] \leq (k + V(N)) \cdot \sigma^2$.
\end{proof}

\noindent\textbf{Remark} (bounded state). At most two containers are active at the same time, since primary and shadow starts are offset by $N/2$ AWs. Each primary tree requires at most $h+1$ running partial-sum counters (one per node on the incomplete root-to-leaf path), and the shadow maintains $1$ running sum, for $O(\log N)$ total raw state.

\section{API~3: Cross-Locality Aggregation -- Intuitive Privacy Analysis}
\label{appendix:api3}

The \sysname API~3 design and privacy guarantees largely inherit from Attribution. We leverage prior formal analyses of Attribution~\cite{TKM+24} via a mapping between API~3 measurements and Attribution primitives. Our argument is currently \emph{intuitive}: we outline this mapping and rely on its correspondence to Attribution, but do not formalize it here (formalization is ongoing).

\subsection{Mapping API~3 to Attribution}
\label{appendix:api3:attribution-mapping}

We detail API~3's functioning to make the mapping between API~3 and
Attribution explicit. \F\ref{fig:api3} illustrates the workflow.

\step{1} In API~3, applications register cross-location measurements at
\sysname nodes along (and optionally within) a target corridor. Each
measurement consists of simple device-side instructions that map directly
to Attribution primitives. For example, in a linger-time measurement,
entry nodes instruct devices to record the current time (analogous to
saving an impression), while exit nodes instruct devices to report the
elapsed time since entry (analogous to attributing a conversion to a
previous impression). Additional instructions within the corridor (e.g.,
recording intermediate locations) correspond to further impression events
from which one or more reports are later generated and budgeted separately.

\step{2} Nodes store registered measurements and \step{3} broadcast them
locally. \step{4} The \sysname mobile app receives these broadcasts and
forwards them to Attribution for execution.

\step{4.1} Devices record small amounts of local state corresponding to
these instructions, analogous to Attribution's impression storage. Upon
a reporting trigger, the device generates an encrypted report using
trajectory-specific functions (e.g., bounded time differences or
L1-bounded histograms over visited locations). These are variations of
Attribution's conversion logic, preserving the bounded-sensitivity
constraints required for correct IDP accounting.

\step{4.2} We reuse Attribution's IDP accounting, extending it to include
API~2 privacy losses (\S\ref{sec:api2}), and expose it as a unified
\sysname component.
Unlike API~2, where privacy is enforced at the node
and accounted on devices, API~3 enforces and accounts for privacy
entirely on-device. Each device maintains an epoch-level budget (e.g.,
weekly), deducts privacy loss upon generating reports, and suppresses
further contributions once the budget is exhausted (e.g., by returning
encrypted null values).

\step{5} Devices send encrypted reports to the application backend,
which \step{6} aggregates reports across users and \step{7} submits them
for secure aggregation (MPC/TEE), as in Attribution.
\step{8} The aggregation service computes noisy aggregates consistent
with the on-device accounting and returns DP results to the application.

A key difference from the web setting is the absence of an explicit
interaction channel through which users are associated with specific
measurements. On the web, this association is inherent: a user visiting
a site and performing an action (e.g., a purchase) establishes a shared
view between browser and querier of which measurement the user
participates in. This information is essential for IDP accounting,
ensuring that privacy loss is charged only for relevant measurements.

In \sysname, we reconstruct this association via node-user coordination:
devices self-identify as relevant to specific measurement contexts based
on broadcast annotations and local sensing, and execute the corresponding
reporting logic. Intuitively, this serves as a physical-world analogue of
the implicit interaction channel between the browser and conversion site
in Attribution.

\subsection{Intuitive Privacy Claim}
\label{appendix:api3:privacy-guarantee}

Under the above mapping, an API~3 measurement can be viewed as an
instance of an Attribution measurement in which:
(i) infrastructure broadcasts induce impression events or reporting
triggers,
(ii) device-side logic generates conversion reports, and
(iii) reporting is governed by on-device individual differential privacy
(IDP) accounting.

We therefore invoke \Thm.~1 of \cite{TKM+24}, and claim that API~3,
through its use of Attribution, satisfies
\emph{(device, epoch)-level individual differential privacy}, meaning
that the inclusion or exclusion of any single device's participation
within an epoch (e.g., a week) has a bounded effect on the distribution
of outputs, as determined by the device's configured privacy budget.

The argument is informal. A formal reduction, currently under development,
requires: (1) formalizing the data, query, and system model of \sysname
API~3 and showing equivalence to Attribution's; (2) formalizing
node-user coordination as having an effect equivalent to the
``report identifier'' in~\cite{TKM+24} -- the logical construct that
was introduced to capture the shared context between device and
conversion site, by which a device ``self-identifies'' to participate
(and incur privacy loss) or not in a measurement; and (3) showing
that our unidirectional node-user coordination protocol introduces
no additional privacy leakage beyond what is already captured in
Attribution.

Despite the lack of a full formalization, we highlight an intriguing
parallel: public sensing must explicitly implement a construct that is
inherent on the web, but which Attribution needed to introduce
conceptually to model that behavior. This reinforces our web analogy --
and its breaking points -- as a useful guide for designing privacy in
public sensing.
\section{Evaluation -- Methodology Details}
\label{appendix:evaluation-details}

This section provides full methodological details for each experiment in \S\ref{sec:evaluation}. Each subsection mirrors the corresponding evaluation subsection in the paper body.

\subsection{API~1: On-Scene}
\label{appendix:eval-details:api1}

\heading{Sniffing attack with $\maxEC$ -- methodology.}
We simulate the sniffing attack from \S\ref{sec:api1}: a mobile
adversary traverses the city and, at each intersection, reads the
current API~1 output buffer, capturing all activity accumulated since
the start of the current EC window. We grant the attacker full buffer
access upon arrival, providing an upper bound on any sniffer's
capability. This evaluation targets small-scale sniffing; large-scale
defenses (anonymization, regulation) are orthogonal and not modeled.

We construct the simulation over Midtown Manhattan (14th--59th Streets,
${\sim}3.5 \times 3.5$\,km, ${\sim}540$ intersections), extracted from
OpenStreetMap and converted to SUMO~\cite{lopez-sumo-2018}. Synthetic
traffic traces are calibrated to 10\% of peak-hour Manhattan
volumes~\cite{nyc-dot-traffic-2024,nyu-wagner-manhattan-2012,nyc-dot-cycling-2024},
providing realistic density in the absence of public city-scale traces.

We model four attacker profiles: pedestrian (1.3\,m/s), cyclist
(5.4\,m/s), car (13.4\,m/s), and a static device at Times Square.
Mobile attackers follow a greedy routing strategy that maximizes new
intersections visited per unit time.

We sweep $\maxEC \in \{30\text{s}, 1\text{min}, 2\text{min}, 3\text{min}, 4\text{min}, 5\text{min}\}$
and report the fraction of total city-wide \emph{person-activity}
(person-seconds across all intersections) captured over one hour.

\heading{Frames dropped with $\maxEC$ -- methodology.}
We evaluate API~1 container-rotation overhead using the Object Detection
and Tracking service (YOLO26n), which processes raw video frames and
emits per-frame bounding boxes. We use the MTID pole-camera dataset
(38 minutes, 30\,FPS)~\cite{jensen-mtid-2020} and measure the percentage
of input frames dropped (not processed within the real-time deadline)
while sweeping $\maxEC \in \{5, 10, 15, 20, 25, 30, 40, 50, 60\}$\,s
on an Nvidia 5060~Ti (desktop GPU) and an Nvidia Jetson AGX Orin (edge GPU).

\begin{figure}[t]
\centering
\begin{subfigure}[t]{0.48\columnwidth}
    \centering
    \includegraphics[width=\linewidth]{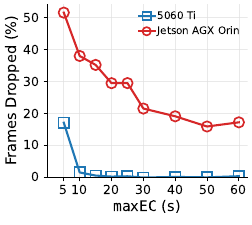}
    \caption{Frames dropped (\%)}
    \label{fig:api1-frames-dropped}
\end{subfigure}
\hfill
\begin{subfigure}[t]{0.48\columnwidth}
    \centering
    \includegraphics[width=\linewidth]{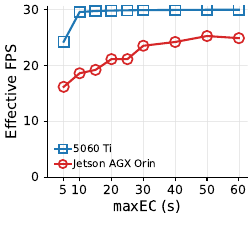}
    \caption{Effective FPS}
    \label{fig:api1-effective-fps}
\end{subfigure}
\vspace{1em}
\caption{{\bf API~1 performance overhead vs.\ $\maxEC$} on a desktop
5060~Ti and an edge Jetson AGX Orin.
(a)~Percentage of input frames dropped during container rotation.
(b)~Effective throughput in frames per second.
The 5060~Ti reaches near-zero drops and 30\,FPS by $\maxEC{\geq}10$\,s;
the Jetson plateaus near ${\sim}17\%$ drop / ${\sim}25$\,FPS due to its
slower GPU.
Each data point runs for $5{\times}\maxEC$ to ensure at least five
rotation cycles.}
\label{fig:api1-latency}
\end{figure}


\subsection{API~2: Single-Locality Aggregation}
\label{appendix:eval-details:api2}

\heading{DP counting accuracy -- methodology.}
We evaluate counting accuracy on NYC DOT bicycle counts from the
Williamsburg Bridge Eco-Counter (15-minute bins over 30 days,
2{,}880 bins).

We implement DP continual release using the Toeplitz
mechanism~\cite{FHU23}, a state-of-the-art approach that
provides stronger utility than the tree-based mechanism described
in \S\ref{sec:api2} while preserving analogous guarantees. We describe
the tree-based mechanism in the body for simplicity of exposition,
but use Toeplitz in our prototype for its superior performance
in both the trusted and untrusted modes.

We aggregate bins into windows ranging from 1 to 12~hours and
add Gaussian noise matching the Toeplitz marginal distribution
($T_{\max} = 1{,}440$, $\delta = 10^{-5}$). We report root mean square relative error (RMSRE)~\cite{hay-boosting-2010} averaged over 1{,}000
Monte Carlo trials. This per-period simulation is unbiased
(\S\ref{appendix:api2:data-model}).

We sweep $\varepsilon \in \{0.1, 1.0\}$ under two configurations:
\emph{trusted} (sensitivity $\Delta = 5$, justified below) and
\emph{untrusted} (output clamping to $[0,100]$, yielding $\Delta = 100$).
The latter value is chosen to match the order of magnitude of counts over
the largest aggregation window (12~h), ensuring that measurements are not
artificially clamped in the untrusted setting. We assume the querier does
not know the exact value a priori, but can set a reasonable
order-of-magnitude upper bound, here $100$. The sensitivity ratio
$\Delta_U / \Delta_T = 20$ directly scales the noise. The hero figure
plots three lines -- trusted at $\varepsilon{=}1$, trusted at
$\varepsilon{=}0.1$, and untrusted at $\varepsilon{=}1$ -- to
highlight that trusted mode more than compensates for a $10{\times}$
tighter privacy budget.

\heading{Empirical sensitivity benchmark for trusted mode -- methodology.}
\headingi{Pipeline under test.}
We run the identical detector and tracker configuration deployed by the
\sysname{} Object Detection and Tracking service.
The detector is YOLOv26n~\cite{citation-needed-yolo26n}, an NMS-free
end-to-end architecture, with confidence threshold 0.5 and person-only
class filter (class~0).
The tracker is BoT-SORT~\cite{aharon-botsort-2022} using the production
configuration: high/low detection thresholds of 0.25/0.1, track buffer of
90~frames (3\,s at 30\,fps, extended from default to accommodate
\sysname{}'s frame-dropping architecture), IoU matching threshold 0.8,
no camera motion compensation (\texttt{gmc\_method: none}), and no
re-identification network (\texttt{with\_reid: false}).
A fresh model instance is created for each sequence, resetting all
tracker state (track IDs, Kalman filters, lost track buffer) to match
\sysname{}'s per-ephemeral-context lifecycle.

\headingi{Dataset.}
We use the MOT17 benchmark~\cite{milan-mot16-2016} training set, which
provides annotated pedestrian sequences at $1920{\times}1080$ resolution.
We restrict to the four static-camera sequences (MOT17-02, 04, 05, 09;
3{,}012~frames total), excluding three moving-camera sequences
(MOT17-10, 11, 13) because \sysname{} deploys on static pole-mounted
cameras and the BoT-SORT configuration disables camera motion
compensation. Ground truth is filtered to class~1 (pedestrian) with
confidence flag~1 (not ignored).

\headingi{Method.}
For each frame, ground-truth bounding boxes are matched to tracker
outputs using the Hungarian algorithm~\cite{kuhn1955hungarian} on an IoU
cost matrix. We discard assignments with $\text{IoU} < 0.5$, following
the CLEAR MOT standard~\cite{bernardin-clearmot-2008}. For each GT
identity across all frames in a sequence, we collect the set of distinct
tracker IDs matched to it; the cardinality of this set is the number of
track IDs assigned to that person. GT identities never matched to any
tracker output are excluded to avoid conflating detection recall with
tracking ID consistency.

\headingi{Limitations.}
MOT17 depicts European street scenes; deployment may involve different
camera angles, lighting, or crowd densities. The sequences are
17--35\,s at 30\,fps; longer sequences may exhibit more ID switches as
individuals leave and re-enter the scene beyond the 3\,s track buffer.
The confidence threshold of 0.5 matches production but discards
low-confidence detections, which may cause track fragmentation on
re-detection. BoT-SORT runs without ReID, so re-identification after
long occlusions relies solely on Kalman filter motion prediction;
enabling ReID would likely reduce $\Delta$ but is not part of the
production configuration.

\subsection{API~3: Cross-Locality Aggregation}
\label{appendix:eval-details:api3}

\heading{Subway OD estimation under per-device budgets -- methodology.}
We use the Hangzhou Metro smart-card dataset~\cite{liu-hangzhou-metro-2019},
containing 29.2\,M trips from 4.97\,M riders across 80 stations over 25 days.
Unlike aggregate OD datasets, it provides per-user histories, enabling
direct modeling of budget exhaustion.

Entering station $A$ records an event locally; exiting at $B$ triggers
a report consuming $\varepsilon_{\text{rep}} = 0.5$ from the device's
weekly budget $\varepsilon$. Once $\varepsilon$ is exhausted, subsequent
reports are replaced with encrypted nulls, introducing bias due to
missing contributions. We capture this alongside standard DP noise in
RMSRE, where budget exhaustion contributes bias error and DP noise
contributes unbiased variance.

Each station aggregates received reports into an $80 \times 80$ OD
matrix and adds Laplace noise with scale $b = 2.0$.

We vary two parameters: (1) device budget $\varepsilon$ (1--10 and
$\infty$), and (2) batch length (hour/day/week). RMSRE is computed
exactly (no Monte Carlo) since both budget exhaustion and noise are
analytically characterized.


\subsection{API 3: Self-Identification}
\label{appendix:eval-details:self-idenfification}

\heading{Stage~1: IMU activity classification -- methodology.}
We employ a histogram-based gradient boosting classifier to distinguish the three modes of human activity: walking, driving, and cycling based on the six-channel IMU data (accelerometer and gyroscope). Raw recordings are obtained from three heterogeneous sources: Vi-Fi pedestrian traces ~\cite{vi-fi_2022}, in-vehicle accelerometer and gyroscope pairs ~\cite{driver-2022}, and PAMAP2 cycling segments ~\cite{pamap2}. All samples go through a shared preprocessing stage that includes normalization to rad/s, resampling frequencies, and segmenting the continuous signals to 3-second sliding windows. Each window is flattened into a 900-dimensional feature vector and fed to a histogram-based gradient boosting classifier in a stratified group k-fold setting to prevent data leakage. 

To assess the robustness of the trained classifier to sensor noise, we evaluated its performance under synthetically injected Gaussian perturbations of increasing magnitude. For each feature, we computed the standard deviation across the held-out test set. At each noise level $\sigma_\text{scale}$, zero-mean Gaussian noise was drawn independently for each observation and each feature dimension, scaled by the product of the per-feature standard deviation and $\sigma_\text{scale}$, and added to the clean test data. A noise level of 0.0 corresponds to the unperturbed test set, while a noise level of 1.0 indicates that the injected perturbation has the same spread as the natural variability of each feature. This feature-relative scaling ensures
that noise magnitudes are comparable across features with different dynamic ranges, providing a realistic simulation of sensor degradation. Multiple independent trials were averaged at each noise level to reduce sampling variance.

\heading{Stage~2: Bounding-box to IMU matching -- methodology.}
Stage~2 matches each phone to a visual bounding-box track.
For each 5\,s window (2.5\,s hop), we
compute Welch PSDs of IMU magnitude and visual motion signals in the
0.7--3.0\,Hz walking band, form a $P{\times}T$ cost matrix of pairwise
cosine distances, augment it to $(P{+}T){\times}(P{+}T)$ with an
adaptive unassignment penalty $\lambda$ (75th-percentile of the cost
distribution), and solve the assignment with the Hungarian
algorithm~\cite{kuhn1955hungarian}.  Any phone or track matched to a
dummy counterpart is declared unmatched.

We evaluate on Vi-Fi outdoor sequences~\cite{vi-fi_2022} using window-level micro-averaged F1-score. In deployment, the camera observes many tracks without a matching phone. We inject 0--10 distractor tracks from other scenes (3 random trials).

\heading{Impact on API~3 utility -- methodology.}
We reuse the Hangzhou Metro trip dataset from \S\ref{sec:eval:api3}
(29.2\,M trips, 80~stations, 25~days) and enrich it with transfer
information via shortest-path routing over the station adjacency graph,
identifying 13.6\,M transfer events -- passengers present at
intermediate stations without exiting.

We parameterize the end-to-end self-identification pipeline by a single
detection rate~$a \in [0,1]$: recall (fraction of true exiters correctly
detected)~$= a$, and false positive rate (fraction of transfer
passengers falsely detected as exiters)~$= 1{-}a$. This single-parameter
sweep suffices for the sensitivity analysis we present. Stage~2
(cross-modal matching) determines \emph{which} bounding-box track a
phone corresponds to, not \emph{whether} an exit occurred, so its errors
are orthogonal to the OD error modeled here. A second parameter
$p \in \{0, 0.5, 1.0\}$ scales the fraction of transfer
passengers near exit gates, modeling station geometries from transfers
far from gates ($p{=}0$) to fully overlapping transfer and exit paths
($p{=}1$).

For each $(a, p)$ pair, the expected per-bin signed OD error is
$(a{-}1) \cdot n_{\text{true}} + (1{-}a) \cdot p \cdot
n_{\text{transfer}}$: missed exits (negative) and spurious entries
(positive). We aggregate over weekly batches and compute RMSRE as in
\S\ref{sec:eval:api2}, including the analytic Laplace noise floor
($b = 1/\varepsilon_{\text{rep}} = 2.0$, yielding
RMSRE${\,=\,}0.0025$). We report results against F1 score; because
$p$~affects the false positive count, F1 depends on both~$a$ and~$p$,
so curves for different~$p$ values occupy different x-axis ranges.

\end{document}